\documentclass[letterpaper, 10 pt, conference]{ieeeconf}  % Comment this line out
\IEEEoverridecommandlockouts                              % This command is only
\usepackage[utf8]{inputenc}
\usepackage{amssymb,amsmath}
\usepackage{graphicx}
\usepackage{hyperref}       % hyperlinks
\usepackage{url}            % simple URL typesetting
\usepackage{booktabs}       % professional-quality tables
\usepackage{amsfonts}       % blackboard math symbols
\usepackage{nicefrac}       % compact symbols for 1/2, etc.
\usepackage{microtype}      % microtypography
\usepackage{xcolor}
\usepackage{amsfonts,amsmath, amssymb}
\usepackage{bbm}
\usepackage[ruled,vlined,linesnumbered]{algorithm2e}
\DontPrintSemicolon
\usepackage{thmtools, thm-restate}
\usepackage{graphicx}
\usepackage{doi}
\usepackage{bm}
\usepackage{tabularx}
\usepackage{adjustbox}
\usepackage{mathtools}
\usepackage{tikz,color,chngpage,soul,hyperref,csquotes,graphicx,floatrow, yfonts}
\usepackage{cite}

\usepackage[capitalise,nameinlink]{cleveref}    
\crefformat{equation}{#2(#1)#3}
\crefformat{figure}{#2Fig.~#1#3}
\crefformat{section}{#2\S#1#3}
\crefformat{subsection}{#2\S#1#3}
\crefformat{subsubsection}{#2\S#1#3}
\crefformat{assumption}{#2Assumption~#1#3}
\crefformat{assumption}{#2Assumption~#1#3}
\renewcommand{\baselinestretch}{0.975}
\usepackage[subtle]{savetrees} %
\DeclareMathOperator*{\minimize}{\mathrm{minimize}}
\DeclareMathOperator*{\subjectto}{\mathrm{subject~to}}

\title{\LARGE \bf
Closing the Loop on Runtime Monitors with Fallback-Safe MPC
}

\author{Rohan Sinha, Edward Schmerling, and Marco Pavone% <-
\thanks{The authors are with the Dept. of Aeronautics and Astronautics at Stanford University, Stanford, CA. \texttt{\{rhnsinha, schmrlng, pavone\}@stanford.edu}. The NASA University Leadership initiative (grant \#80NSSC20M0163) provided funds to assist the authors with their research, but this article solely reflects the opinions and conclusions of its authors and not any NASA entity.}%
}

\begin{document}

\maketitle
\thispagestyle{empty}
\pagestyle{empty}

%%%%%%%%%%%%%%%%%%%%%%%%%%%%%%%%%%%%%%%%%%%%%%%%%%%%%%%%%%%%%%%%%%%%%%%%%%%%%%%%
\begin{abstract}
When we rely on deep-learned models for robotic perception, we must recognize that these models may behave unreliably on inputs dissimilar from the training data, compromising the closed-loop system's safety. This raises fundamental questions on how we can assess confidence in perception systems and to what extent we can take safety-preserving actions when external environmental changes degrade our perception model's performance. Therefore, we present a framework to certify the safety of a perception-enabled system deployed in novel contexts. To do so, we leverage robust model predictive control (MPC) to control the system using the perception estimates while maintaining the feasibility of a safety-preserving fallback plan that does not rely on the perception system. In addition, we calibrate a runtime monitor using recently proposed conformal prediction techniques to certifiably detect when the perception system degrades beyond the tolerance of the MPC controller, resulting in an end-to-end safety assurance. We show that this control framework and calibration technique allows us to certify the system's safety with orders of magnitudes fewer samples than required to retrain the perception network when we deploy in a novel context on a photo-realistic aircraft taxiing simulator. Furthermore, we illustrate the safety-preserving behavior of the MPC on simulated examples of a quadrotor. We open-source our simulation platform and provide videos of our results at our project page: \url{https://tinyurl.com/fallback-safe-mpc}.
\end{abstract}

\newcommand{\calZ}{\mathcal{Z}}
\newcommand{\calY}{\mathcal{Y}}
\newcommand{\calX}{\mathcal{X}}
\newcommand{\calW}{\mathcal{W}}
\newcommand{\calV}{\mathcal{V}}
\newcommand{\calU}{\mathcal{U}}
\newcommand{\calT}{\mathcal{T}}
\newcommand{\calS}{\mathcal{S}}
\newcommand{\calR}{\mathcal{R}}
\newcommand{\calQ}{\mathcal{Q}}
\newcommand{\calP}{\mathcal{P}}
\newcommand{\calO}{\mathcal{O}}
\newcommand{\calN}{\mathcal{N}}
\newcommand{\calM}{\mathcal{M}}
\newcommand{\calL}{\mathcal{L}}
\newcommand{\calK}{\mathcal{K}}
\newcommand{\calJ}{\mathcal{J}}
\newcommand{\calI}{\mathcal{I}}
\newcommand{\calH}{\mathcal{H}}
\newcommand{\calG}{\mathcal{G}}
\newcommand{\calF}{\mathcal{F}}
\newcommand{\calE}{\mathcal{E}}
\newcommand{\calD}{\mathcal{D}}
\newcommand{\calC}{\mathcal{C}}
\newcommand{\calB}{\mathcal{B}}
\newcommand{\calA}{\mathcal{A}}

\newcommand{\todo}[1]{{\color{red}{TODO: #1}}}

\newcommand{\E}{\mathbb{E}}
\newcommand{\p}{\mathbb{P}}
\newcommand{\R}{\mathbb{R}}
\newcommand{\Ps}{\mathcal{P}}

\newcommand{\calhA}{\widehat{\mathcal{A}}}
\newcommand{\calhU}{\widehat{\mathcal{U}}}

\newcommand{\Qhat}{\widehat{Q}}
\newcommand{\Vhat}{\widehat{{V}}}

\newcommand{\prob}{\mathrm{Prob}}%{\mathbb{P}}
\newcommand{\iid}{\overset{\textup{iid}}{\sim}}

\newcommand{\pe}{\bm{e}}
\newcommand{\x}{\bm{x}}
\newcommand{\y}{\bm{y}}
\newcommand{\z}{\bm{z}}
\newcommand{\f}{\bm{f}}
\newcommand{\g}{\bm{g}}
\newcommand{\h}{\bm{h}}
\newcommand{\w}{\bm{w}}
\newcommand{\vn}{\bm{v}}
\newcommand{\fu}{\bm{u}}
\newcommand{\traj}{\bm{\tau}}
\newcommand{\params}{\bm{\theta}}

\newcommand{\pre}{\mathrm{Pre}}

\newcommand{\nn}{\mathrm{NN}}

\newtheorem{claim}{Claim}
\newtheorem{remark}{Remark}
\newtheorem{lemma}{Lemma}
\newtheorem{theorem}{Theorem}
\newtheorem{corollary}{Corollary}
\newtheorem{definition}{Definition}
\newtheorem{assumption}{Assumption}
\newtheorem{example}{Example}
%%%%%%%%%%%%%%%%%%%%%%%%%%%%%%%%%%%%%%%%%%%%%%%%%%%%%%%%%%%%%%%%%%%%%%%%%%%%%%%%
\section{INTRODUCTION}
Autonomous robotic systems increasingly rely on machine learning (ML)-based components to make sense of their environment. In particular, deep-learned perception models have become indispensable to extract task-relevant information from high-dimensional sensor streams (e.g., images, pointclouds). However, it is well known that modern ML systems can behave erratically and unreliably on data that is dissimilar from the training data --- inputs commonly termed out-of-distribution (OOD) \cite{SinhaSharmaEtAl2022,GeirhosJacobsenEtAl2020, TorralbaEfros2011}. During deployment, ML-enabled robots inevitably encounter OOD inputs corresponding to edge cases and rare, anomalous scenarios \cite{SinhaSharmaEtAl2022, SeshiaSadighEtAl2016}, which pose a significant safety risk to ML-enabled robots.
% While there has been much progress on algorithms that improve the overall (e.g., mean) performance of ML-enabled components under distributional shift in recent years, this is not sufficient to avoid safety critical failures: Edge cases and rare, anomalous scenarios will always persist, and we must ensure that we detect inference errors on individual inputs to avoid robot failures. As in \cite{SinhaSharmaEtAl2022}, we refer to this viewpoint --- that we must characterize when and where we can have confidence in a learned model's predictions --- as the \emph{functional uncertainty} perspective on OOD data.

In this work we examine vision-based control settings, where we rely on a deep neural-network (DNN) to extract task-relevant information from a high-dimensional image observation. When the DNN fails, access to this information is lost, so that we can no longer estimate the full state.
For example, the drone in \cref{fig:hero-fig} relies on a DNN for obstacle detection and subsequent avoidance, and the aircraft in \cref{fig:example} utilizes a DNN to estimate its runway position for tracking control.
Because failures caused by OOD data are difficult to anticipate, recent years have seen much progress on algorithms that monitor the performance of ML-enabled components at runtime \cite{RahmanCorkeEtAl2021,SalehiMirzaeiEtAl2021,RuffKauffmanEtAl2021}. These OOD detection algorithms aim to detect inference errors so that downstream safety-preserving interventions may be adopted. 
% However, as illustrated in \cref{fig:hero-fig} and \cref{fig:example}, when ML-based components fail, this may result in an irretrievable loss of information, without which it is impossible to estimate the state.
However, few works attempt to integrate such monitors into a perception and control stack. Instead, existing work typically assumes that estimation errors will always satisfy known bounds or that there exists a safe fallback that can always be triggered under loss of sensing. 
To derive end-to-end certificates on the safety of the monitor-in-the-loop system, two key challenges must be addressed: (1) an OOD detector must be calibrated to detect violations of assumptions underpinning the nominal control design and (2) the control strategy itself must be aware of the limitations of a safety-preserving intervention. This latter challenge is of particular importance for safety, since, as illustrated in \cref{fig:hero-fig}, naively executing a specified fallback (e.g., landing the drone) can introduce additional hazards.

% However, many practical OOD detectors are DNNs themselves, and as a result, it is rarely possible to construct end-to-end certificates on the safety of a perception and control stack with a runtime monitor based on existing work. 
% Therefore, we jointly consider the problems of 1) constructing a runtime monitor that certifiably detects perception faults and 2) safely transitioning the robot into a degraded operational state when we detect ML-enabled components become unreliable.

\begin{figure}
    \centering
    \includegraphics[width=\textwidth]{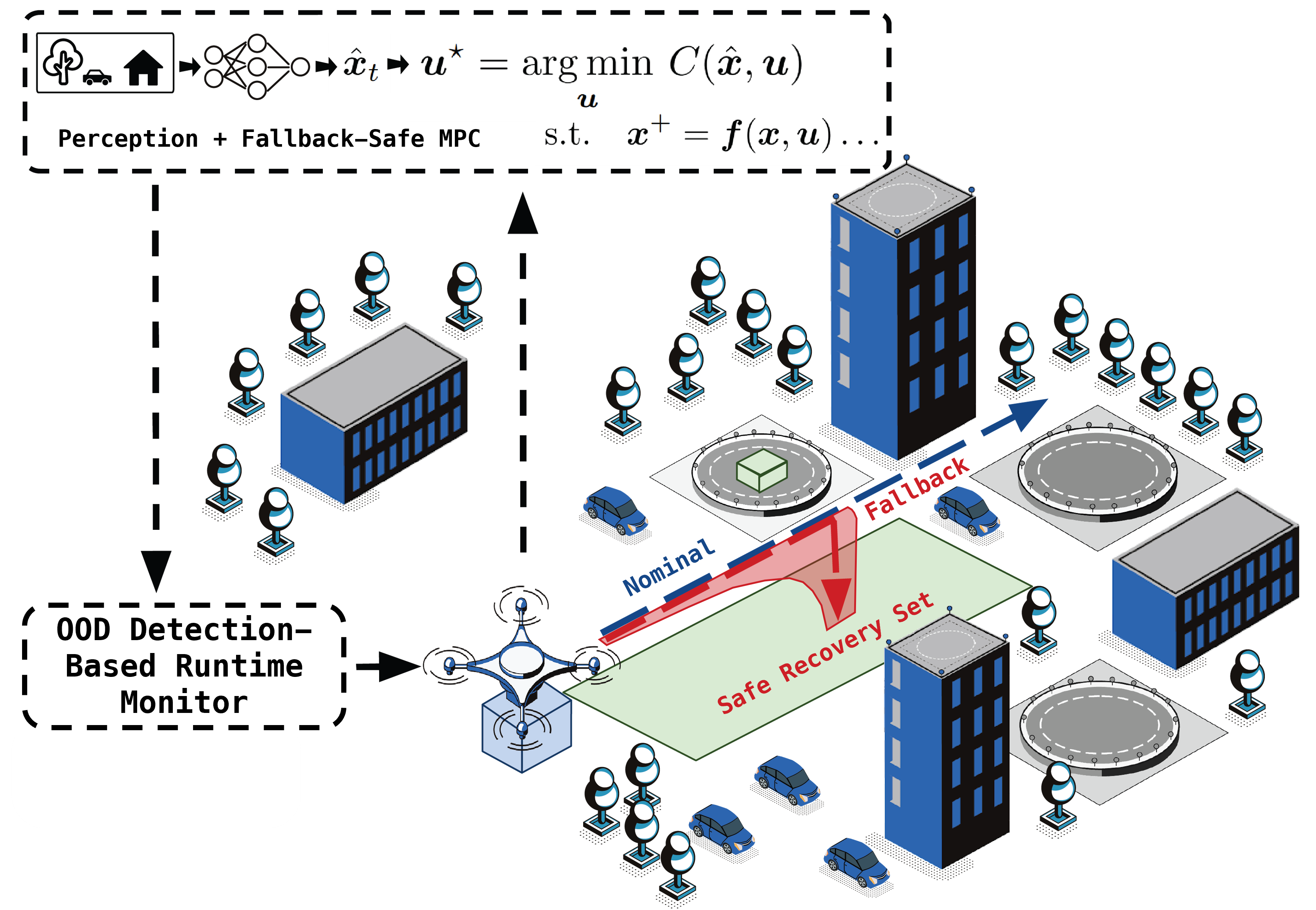}
    \caption{Overview of the proposed approach: A drone delivery service uses camera-vision to navigate around a city. However, the ML perception system behaves unreliably on out-of-distribution (OOD) inputs. Therefore, we construct a runtime monitor to trigger a fallback strategy when the perception system is unreliable. To do so, we calibrate heuristic OOD detection scores to decide when to land the drone. To operate this fallback safely, we must ensure that it does not drop down into trees or roads. Therefore, we minimally modify the nominal system operation to ensure we fallback into a safe recovery set using a robust MPC.}
    \label{fig:hero-fig}
\end{figure}
\begin{figure*}[t!]
  \centering
  \includegraphics[width=.9\textwidth]{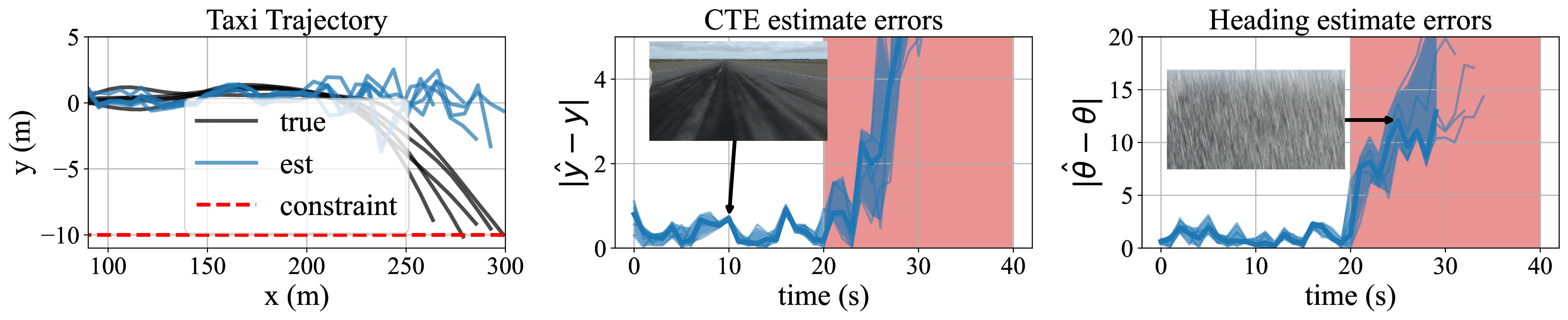}
  \caption{20 simulated trajectories of an autonomous aircraft taxiing down a runway. A DNN -- trained on data collected in morning, clear-sky weather -- estimates the cross-track error (CTE) and the heading error (HE) used by a tracking controller from vision. At the start of the trajectories, the weather is clear. At $t=20s$, it starts raining. Leftmost plot: $x$ is the down-track position along the runway, and $y$ is the CTE. 
  % The DNN estimate is initially adequate but once it sees OOD data, diverges rapidly from the true trajectory, resulting in failures (running off the runway) in all cases. 
  Middle and Rightmost plots: Errors of the DNN perception cross-track error and heading error estimate. In both cases, the estimate is initially of sufficient quality. However, when it starts raining, estimates immediately degrade (OOD time steps highlighted in red), causing the robot to fail (running off the runway) for all trajectories. \vspace{-2em}}
  \label{fig:example}
\end{figure*}
To address these challenges, we present the Fallback-Safe Model Predictive Control (MPC) framework which, while maintaining safety, aims to derive maximum utility from the DNN component necessary for nominal task success. Our framework satisfies three key desiderata associated with the above challenges, namely: (1) we ensure the safety of the fallback strategy, i.e., we do not assume the existence of a ``catch-all'' fallback, (2) we explicitly quantify DNN and runtime monitor performance to inform control without resorting to conservative worst-case assumptions and (3) in the context of robust control, we account for the existence of errors, i.e., DNN failures, of arbitrary severity through the development of a runtime monitor.

In this framework, we first specify an error bound on the quality of state estimates produced by the ML perception system when it operates in-distribution (i.e., in nominal conditions). In nominal conditions, we robustly control the system with respect to the specified perception error bound. Then, we calibrate an OOD detection heuristic to trigger a safety-preserving fallback strategy when the chosen perception error bound is violated, resulting in an end-to-end guarantee that the system will satisfy state and input constraints with high probability, regardless of DNN failure. This calibration procedure does require examples of such OOD cases, though notably the strength of our guarantees scales with calibration dataset size, irrespective of DNN complexity (as would, e.g., retraining on OOD cases). 
% In our view, availability of some OOD data is reasonable when we seek to deploy an autonomous 
As such, the setting we consider captures the practice of running a few pre-deployment trials when we want to deploy a pre-trained system in a new context, which we know differs from the training distribution. For example, we can anticipate shifted conditions and run some calibration trials when we deploy an autonomous aircraft trained on images of American runways in Europe or when we deploy an image classifier trained on ImageNet as a part of our autonomy stack.

\textbf{Related Work:} Most similar to our approach are several works that consider triggering a fallback controller by thresholding an OOD detection algorithm \cite{RichterRoy2017, FilosTigas2020, McAllisterKahn2019}. These works use fallback strategies that are domain-specific or naively assumed to always be safe, like a hand-off to a human, and some approaches assume the ML models function perfectly nominally \cite{RichterRoy2017}. Moreover, to detect OOD conditions, they either rely on additional DNNs for OOD detection \cite{RichterRoy2017, McAllisterKahn2019}, or use approximate techniques to quantify uncertainty \cite{FilosTigas2020}, and thus do not make strong guarantees of system safety. 
% Moreover, as illustrated in \cref{fig:hero-fig}, naively executing a specified fallback (e.g., landing the drone) can create additional hazards. 
Similar to existing work on fault-tolerant control that maintains feasibility of passive-backup \cite{GuffantiAmico2023}, abort-safe \cite{MarsillachCairano2020}, or contingency plans \cite{AlsterdaBrownEtAl2019} under actuation or sensor failure using MPC, we modify the nominal operation to ensure the existence of a safety preserving fallback. However, in many such works, it is assumed that faults are perfectly detected and that these systems function perfectly nominally. It is challenging to detect failures in ML-based systems and, as illustrated in \cref{fig:example}, errors are nominally tolerable, but nonzero. Therefore, we jointly design the control stack and runtime monitor to account for such errors. 

%talk about safety filters and points one and two

Our approach takes inspiration from \emph{safety filters} as defined in \cite{BrunkeGreefEtAl2022}. Such methods minimally modify a black-box policy's actions to ensure invariance of a safe subset of the state space when the black-box policy would take actions to leave that set. Many such algorithms have been proposed in recent years, for example by using robust MPC \cite{WabersichZeilinger2018}, control barrier functions (CBFs) \cite{ChengOrosz2019} and Hamilton-Jacobi reachability analysis \cite{FisacAkametaluEtAl2019, LeungSchmerlingEtAl2019}. However, our method differs from such approaches in two important ways:
First, existing safety filters operate under assumptions of perfect state knowledge, or that an estimate of known accuracy is always available \cite{BrunkeZhou2022}. However, ML-enabled components for perception are necessary to complete the control task in many applications. When these components fail, it becomes impossible to estimate the full state (e.g., see \cref{fig:hero-fig}, \cref{fig:example}). 
% For example, the drone in \cref{fig:hero-fig} simply cannot avoid obstacles without obstacle detectors, and the autonomous aircraft in \cref{fig:example} cannot estimate its position on the runway without its vision system. 
We account for these discrete information modes by defining \emph{recovery sets} to fallback into, which formalize the intuition that e.g., the drone in \cref{fig:hero-fig} does not need an obstacle detector to avoid mid-air collisions when it is landed in a field. 
% We account for these discrete information modes, our insight is that we can often specify \emph{recovery sets} to fallback into, safe subsets of the state space that we can make invariant without full knowledge of the state. For example, the drone in \cref{fig:hero-fig} does not need an obstacle detector to avoid mid-air collisions when it is landed in a field. 
Secondly, existing safety filters take a zero-confidence view in black-box learned components: They continually ensure the safe operation of the system by aligning an ML-enabled controller's output with those of a backup policy that never uses ML models. Instead, we recognize that ML-enabled components are generally reliable, but leverage OOD detection to transition to a fallback strategy in rare failure modes.

%mention control robust to perception and output feedback, and point 3
Robust control from imperfect measurements, output-feedback, classically relies on state estimators that persistently satisfy known error bounds, constructed using known measurement models with assumptions on system observability. In particular, output-feedback MPC controllers robustly satisfy state and input constraints for all time (e.g., see \cite{MayneRakovic2006,LorenzettiPavone2020,LovaasSeronEtAl2008,KohlerMuller2021,FindeisenImslandEtAl2003,GoulartKerrigan2006}). Typically, these methods robustly control the state estimate dynamics using a robust MPC algorithm, tightening constraints to account for the estimation error \cite{MayneRakovic2006, LorenzettiPavone2020, KohlerMuller2021, GoulartKerrigan2006}. However, we cannot model high-dimensional measurements like images from first principles, and as a result, we must rely on neural networks to extract scene information. Some recent approaches propose to learn the error behavior of a vision system as a function of the state and 
% then verify closed-loop properties \cite{KatzCorso2022} or
robustly plan while taking these error bounds into account \cite{DeanMatni209, ChouOzay2023, IchterLandry2020}, for example, by making smoothness assumptions on the vision's error behavior \cite{DeanMatni209}. These algorithms require the environment to remain fixed (i.e., that the mapping from state to image is constant). However, oftentimes external changes to the environment, like the changing weather in \cref{fig:example}, result in OOD inputs that cause arbitrarily poor perception errors.
% conditions that degrade perception systems are oftentimes externally caused by a changing environment,
% like the changing weather in \cref{fig:example}, and perception systems may exhibit arbitrarily poor error behavior on OOD inputs. 
We rely on the ideas of output-feedback MPC to nominally control the robot, but maintain feasibility to a backup plan in case the perception unexpectedly degrades.   

%justify conformal prediction
To make an end-to-end guarantee, we leverage recent results in conformal prediction to learn how to rely on a heuristic OOD detector (see \cite{AngelopoulosBates2022} for an overview). Conformal methods are attractive in this setting because they produce strong guarantees on the correctness of predictions, are highly sample efficient, and the guarantees are distribution-free -- that is, they do not depend on assumptions on the data-generating distribution \cite{AngelopoulosBates2022, BalasubramanianHoEtAll2014}.
% To train the predictor, we require a small number of trajectory rollouts with ground truth information, which is  
% Therefore, o pre-deployment trials with a safety driver when we want to deploy a pretrained system in a new context. 
% Even though some failure modes may be ``known,'' in the sense that we can collect some data on them, conformal algorithms are accurate with orders of magnitude less data than required to retrain the perception system \cite{LuoZhao2023}. 
% Furthermore, it can be highly non-obvious how to improve the perception stack to fix known failure modes.
However, existing conformal prediction methods do not make high-probability guarantees jointly over predictions along a dynamical system's trajectory, where inputs are highly correlated over time. While some recent work has aimed to move beyond the exchangeable data setting \cite{BarberCandes2023, TibshiraniBarber2019}, or makes sequentially valid predictions across exchangeable samples \cite{LuoSinha2023}, these methods cannot be applied sequentially over correlated observations within a trajectory. Instead, we adapt an existing algorithm, \cite{LuoZhao2023}, to yield a guarantee jointly over the repeated evaluations of the predictor within a trajectory. 
% noticing we only require a guarantee on the first detection of a fault to ensure the end-to-end safety of our framework.

% There is a vast landscape of existing work that explores how to cope with OOD data (e.g., see \cite{SinhaSharmaEtAl2022} for an overview), not all of which we can discuss here. We briefly review the approaches most closely related to our approach.
% We include a more extensive review of existing work, further experimental details, and all proofs of our theoretical results in an extended version, available at \cite{SinhaSchmerlingEtAl2023}. 
For a more extensive review of existing work (including additional approaches for coping with OOD data), see \cref{ap:relwork}.

\textbf{Contributions:} In brief, our contributions are that
\begin{enumerate}
    \item First, we introduce the notion of the safe recovery set, a safe subset of the state space that is invariant under a recovery policy that does not rely on the ML-enabled perception or a full state estimate.
    %which formalizes the intuitive idea that there often
    % formalize the intuitive idea by introducing the notions of recovery policies and recovery regions, allowing us to answer questions like “what makes a fallback ‘good’?” and “does engaging the fallback at the current state create additional hazards?” 
    \item Second, we develop a framework to synthesize a fallback controller and modify the nominal operation of the robot to ensure the existence of a safe fallback strategy for all time by planning into a recovery set.
    \item Third, we propose a conformal prediction algorithm that calibrates the OOD detector, resulting in a runtime-monitor-in-the-loop framework for which we make an end-to-end safety guarantee. 
    % We do so by adapting an existing conformal inference algorithm \cite{LuoZhao2023}, which only guarantees coverage on individual samples, to a sequential decision-making setting. 
    \item Fourth, we evaluate our approach on vision-based aircraft taxiing simulations, where we find that our framework guarantees safety with a low false positive rate on OOD scenarios with orders of magnitudes fewer samples than we used to train the perception system. In addition, we open-source our simulation platform, based on the photorealistic X-Plane 11 simulator (available at: \url{https://tinyurl.com/fallback-safe-mpc}).
\end{enumerate}
% \todo{cut organization?}
Our paper is organized as follows: We outline our problem formulation in \cref{sec:problem}, and then construct our approach in \cref{sec:planning}-\cref{sec:conformal}. We first present the Fallback-Safe MPC framework in \cref{sec:planning} and present the conformal calibration algorithm in \cref{sec:conformal}. Finally, we present simulated results in \cref{sec:simulations}. We include an overview of all notation used in this paper in \cref{ap:glos}.

\section{PROBLEM FORMULATION}\label{sec:problem}
We consider discrete time dynamical systems %systems characterized by discrete time dynamics
\begin{equation}\label{eq:dynamics}
\begin{aligned}
    \x_{t+1} &= \f(\x_t, \fu_t, \w_t),\\
    (\w_t, \z_t) &\sim p_{\rho}(\w_t, \z_t | \w_{0:t-1}, \z_{0:t-1}, \x_{0:t}),
\end{aligned}
\end{equation}
where $\x_t\in \R^n$ is the system state, $\fu_t \in \R^m$ is the control input, $\w_t \in \calW \subseteq \R^d$ is a disturbance signal contained in the known compact set $\calW$ for all time steps $t\geq 0$, and $\z_t \in \R^o$ is a high dimensional observation consisting of an image and more conventional measurements (e.g., GPS), such that $o \gg n$. The joint distribution $p_\rho$ from which disturbances and observations are sampled during an episode depends on an unobserved environment variable $\rho \sim P_\rho$, drawn once at the start of each episode from an environment distribution $P_\rho$.
% which determines the joint distribution $p_{\rho}$ of the process that generates disturbances and observations.

% Crucially, we do not assume the state $\x_t$ is known for all time, but we sense high dimensional observations $\z_t \in \R^o$, like combinations of images and more conventional measurements like GPS readings, such that $o >> n$, at each time step. We assume we deploy the system in a new context, represented as a distribution over environments $\rho \sim P_\rho$, where the environment $\rho$, which is fixed over a trajectory, controls the joint distribution of the process that generates disturbances and observations. That is, disturbances and observations are sampled from some distribution $(\w_t, \z_t) \sim p_{\rho}(\w_t, \z_t | \w_{0:t-1}, \z_{0:t-1}, \x_{0:t})$ over a trajectory. For example, as considered in \cref{sec:simulations}, $\rho$ can represent the time-of-day or weather.
% For generality, we make no further assumptions on the disturbance signal. For example, $\w_t$ may represent 1) \emph{stochastic} dynamics in the form of process noise to capture imperfect knowledge of $\f$ or inherently random dynamics, like i.i.d. process noise, 2) \emph{fixed, but unknown constants}, like unknown parameters of the dynamics, 3) \emph{adversarial perturbations}, chosen by an attacker to maximally disrupt the system, or 4) \emph{the decisions of other agents}, like the actions taken by another agent acting according to a history dependant policy $\w_t = \pi(\x_{0:t}, \fu_{0:t})$. 
Our goal is to ensure the system satisfies state and input constraints over trajectories of finite duration.
\begin{definition}[Safety]\label{def:safety}
Under a risk tolerance $\delta \in [0,1]$ and time limit $t_{\mathrm{lim}} \in \mathbb{N}_{\geq0}$, the system~\cref{eq:dynamics} is \emph{safe} if
\begin{equation}
    \prob(\x_t \in \calX, \ \fu_t \in \calU, \ \forall t \in \{0,1,\dots, t_{\mathrm{lim}}\}) \geq 1 - \delta, %[0, t_{\mathrm{lim}}]
\end{equation}
where $\calX \subseteq \R^n$, $\calU \subseteq \R^m$ are state and input constraint sets.
\end{definition}
% To satisfy \cref{def:safety}, 
\noindent Our setting differs from classic output feedback in that the high-dimensional $\z_t$ cannot be used directly for control.
% , since they do not directly constitute partial measurements of the state in a way that we could model from first-principles (e.g., as a linear function of $\x_t$).
Instead, we consider the setting where a black-box perception system generates an estimate of the state at each time step.
\begin{assumption}[Perception]\label{as:percept}
At each time step $t\geq 0$ a perception system produces an estimate of the state $\hat{\x}_t := \mathrm{perception}(\z_{0:t})$.
\end{assumption}
% To operate on high dimensional observations, perception systems typically use a mix of black-box ML or CV systems (like CNNs or SLAM algorithms) and standard sensors (e.g., GPS, IMU) to produce estimates that are mostly reliable when the ML components operate in-distribution, that is, within the scope of the conditions seen at design time. 
Crucially, the environment distribution $P_\rho$ may differ from the distribution that generated the training data for the learned perception components, which means that we will eventually encounter OOD observations on which the learned components behave erratically. When the learned systems fail, only a restricted amount of information (e.g., an IMU measurement) remains accurate for control. Therefore, we do not make any assumptions on the quality of the learned system outputs, but we assume that the remaining information is accurate for all time.

\begin{assumption}[Fallback Measurement]\label{as:fallback-meas}
At each time step $t\geq 0$ we can access a \emph{fallback measurement} $\y_t \in \R^r$ such that
\begin{equation}
    \y_t = \g(\x_t), \quad \forall t \geq 0.
\end{equation}
We assume $\g$ is known and call $\calY := \{\y = \g(\x) \ : \ \x \in \calX\}$ the \emph{fallback output set}. We restrict ourselves to the setting where the system $\{\f, \g\}$ is not observable; perception is nominally necessary. 
\end{assumption}

Since the estimates $\hat{\x}_t$ may become corrupted unexpectedly, we need to monitor the system's performance online with an intent of detecting conditions that degrade its performance. To do so, we assume we can compute a heuristic OOD detector.
% Then, when we detect an anomaly, we can trigger a fallback controller that retains the safety of the system. However, as discussed more expansively in \cref{sec:relwork}, certifiably detecting when the performance of black-box perception systems built on deep neural networks is degraded is largely an intractable problem. Instead, the literature has produced a wide variety of OOD detection algorithms that function as heuristics for the quality of model outputs. We take advantage of these methods in this work. 
\begin{assumption}[OOD/Anomaly Detection]\label{as:ood}
    At each time step, an OOD/anomaly detection algorithm outputs a scalar anomaly signal $a_t = \mathrm{anomaly}(\z_{0:t}) \in \R$ as an indication of the quality of the state estimate $\hat{\x}_t$ (greater values indicate that the detector has lower confidence in the quality of $\hat{\x}_t$).
    % As a convention, we assume $a_t > a_{t'}$ indicates that the detector has lower confidence in the quality of $\hat{\x}_t$ than that of $\hat{\x}_{t'}$ (i.e., time step $t$ is ``more anomalous'' than $t'$).
\end{assumption}
Note that we make no assumptions on the quality of the OOD detector in \cref{as:ood}: Our approach will certify the safety of the closed-loop system regardless of the correlations between $a_t$ and the perception error $\pe_t :=  \x_t - \hat{\x}_t$. However, the conservativeness of our algorithm will depend on the quality of the heuristic $a_t$. Furthermore, this framework readily allows us to incorporate algorithms that provably guarantee detection of perception errors by letting $a_t$ be an indicator on whether the perception system is reliable. This would only simplify the control design procedure we develop in \cref{sec:planning}-\cref{sec:conformal}. 

\section{Fallback-Safe MPC}\label{sec:planning}
% Our approach in the Fallback-Safe MPC framework
% to the formulation outlined in \cref{sec:problem} centers around the idea that we can
We propose to control the system in state-feedback based on the estimates $\hat{\x}_t$ when the system operates in-distribution and use the anomaly signal $a_t$ to monitor when $\hat{\x}_t$ becomes unreliable. Then, if the monitor triggers, we transition to a fallback policy $\pi: \calY \to \calU$ that only relies on the remaining reliable information, i.e., the fallback measurement $\y_t$.

% That is, a fallback policy is a mapping $\pi: \calY \to \calU$. We dub this framework the Fallback-Safe MPC.
% While fallback systems should safely transition a robot into a minimal risk condition, they can sometimes generate additional hazards. For example, when the collision-avoidance perception system of a drone fails, a reasonable fallback is to have the drone land immediately. However, to operate this fallback safely, the planning and control stack must ensure that the drone can safely land: The drone cannot drop down into power lines or land on the road. As another example: A reasonable fallback for an autonomous car is to stop the vehicle. However, it is not sensible to stop the autonomous vehicle on the highway. 
To avoid that a naively executed fallback creates additional hazards, we make two contributions in this section. First, we introduce the notions of \emph{recovery policies} and \emph{recovery sets}, safe subsets of the state space that we can make invariant without full state knowledge. 
Second, we develop a method to synthesize a fallback controller and modify the nominal operation of the robot to ensure the fallback strategy is feasible for all time.

\subsection{Recovery Policies and Recovery Sets}
\cref{def:safety} requires that a safe fallback satisfies state and input constraints for all time, despite the fact that certain elements of the state are no longer observable. Our insight is that while this is not achievable in general, we can often identify subsets of the state space in which the robot is always safe under some $\pi_R: \calY \to \calU$.  
\begin{definition}[Recovery Set, Policy]\label{def:recovery}
    A set $\calX_R\subseteq \calX$ is a \emph{safe recovery set} under a given recovery policy $\pi_R: \calY \to \calU$ if it is a robust positive invariant (RPI) set under the recovery policy. That is, if 
    \begin{equation}
        \f\big(\x, \pi_R(\g(\x)), \w\big) \in \calX_R \quad\forall \w \in \calW, \quad \forall \x \in \calX_R.
    \end{equation}
\end{definition}
For example, consider the quadrotor in \cref{fig:hero-fig}. The set of all states with altitude $z=0$ and velocity $v=0$ forms a recovery set under the recovery policy $\pi_R(\y_t) := 0$. 
% We assume that we are given a recovery policy $\pi_R$ in advance in \cref{sec:planning-mpc}, allowing a designer to flexibly specify fallback strategies.
\cref{def:recovery} differs from typical definitions in output-feedback problems, because output-feedback control designs typically focus on 1) maintaining the system output $\y_t$ of a partially observed system within a set of constraints despite the unobserved dynamics, or 2) respecting constraints on the true state by bounding the estimation error of an observer. In contrast, existence of a recovery policy allows us to persistently satisfy state constraints, even when estimation errors are unbounded on OOD inputs.

% In some sense, we concern ourselves with the reverse: We want to guarantee safety on the true unobservable state, given only a limited set of observations. 
% \begin{assumption}\label{as:recovery}
%     We are given a \emph{recovery policy} $\pi_R: \calY \to \calU$ associated with a recovery region $\calX_R \neq \emptyset$. 
% \end{assumption}

% \cref{as:recovery} is essentially equivalent to the standard assumption in the robust MPC literature -- access to a terminal controller associated with a nonempty RPI set -- that enables guarantees on persistent feasibility and constraint satisfaction \cite{BorrelliBemporadEtAl2017, MayneRakovic2006, GoulartKerrigan2006, LorenzettiPavone2020}, except that we explicitly enforce that the recovery policy only uses fallback measurements $\y_t$. As in the MPC literature, we note that \cref{as:recovery} can be satisfied by designing a simpler controller (e.g., LQR or human-insight as in the drone landing example) and computing $\calX_R$ using existing algorithms for robust invariant set computation (e.g., see \cite{BorrelliBemporadEtAl2017}), since $f(\x, \pi_R(\g(\x), \w)$ is an autonomous system subject to bounded disturbances. 

\subsection{Planning With Fallbacks}\label{sec:planning-mpc}
We now develop the Fallback-Safe MPC framework. First, to ensure that we satisfy safety constraints nominally, we need to define what it means for the perception system to be reliable in-distribution. 
% This is a subjective choice, because the likelihood that the perception's state estimates are exactly equal to the true state is vanishingly low: We need to define ``how bad is too bad.''
Therefore, we choose a parametric compact state uncertainty set as a bound on the quality of the perception system in nominal conditions.

\begin{definition}[Reliability]
Let the \emph{perception error set} $\calE_{\theta} \subseteq \R^n$ be a symmetric compact set so that $0 \in \calE_{\theta}$. We say an estimate is \emph{reliable} when
    \begin{equation}\label{eq:error}
    \pe_t := \x_t - \hat{\x}_t \in \calE_\theta. 
\end{equation}
We say the perception system is \emph{unreliable}, or experiences a \emph{perception fault}, when $\pe_t \not \in \calE_\theta$.
\end{definition}
We explicitly use the subscript $\theta$ in the construction of $\calE_\theta$ to emphasize that choosing which estimates we consider reliable is a hyperparameter. In this work, we take the perception error set as $\calE_\theta = \{\pe \in \R^n \ : \|\mathbf{A} \pe \|_\infty \leq \alpha\}$ for some $\theta = (\mathbf{A}, \alpha)$ consisting of a matrix $\mathbf{A}$ and bound $\alpha \in \R$ (see \cref{sec:simulations}).
The hyperparameter $\theta$ introduces a trade-off between conservativeness in nominal operation and eagerness to trigger the fallback: Note that 
% $\theta$ makes precise what we mean when we say that we consider the perception system reliable in-distribution. As
as $\calE_\theta$ increases in size, the more state uncertainty we must handle as part of the tolerable estimation errors in nominal conditions, increasing nominal conservatism. 
% We will use the anomaly signal $a_t$ to construct a runtime monitor that detects when $\pe_t \not \in \calE_\theta$ with provable guarantees on the false negative rate in \cref{sec:conformal}. Then, we will engage a fallback when we predict that that $\pe_t \not \in \calE_\theta$, thereby maintaining safety. 
If we decrease the size of $\calE_\theta$, the more eager we will be to trigger the fallback. 
% In practice, one should aim to choose $\calE_\theta$ at least to contain the perception error in reasonable conditions.

Then, in nominal conditions, we control the system using the state estimates $\hat{\x}_t$ with robust MPC, minimally modifying the nominal control objective so that there always exists a fallback strategy that reaches a \emph{recovery set} within $T +1 \in \mathbb{N}_{>0}$ time steps. To do so, we optimize two policy sequences: 1) a sequence of parametric fallback policies $\fu_{t:t+T|t}^F \subset \mathcal{P}_F \subset \{\pi : \R^r \to \R^m \}$, which may only rely on the fallback measurement $\y_t$ and 2) a nominal policy sequence $\fu_{t:t+T|t} \subset \calP_N \subseteq \{\pi: \calX \to \calU\}$ within a state-feedback policy class $\calP_N$, which we assume respects input constraints.
% This notation is commonplace in the MPC literature to admit both optimization over open-loop inputs and (pre-specified) linear feedback gains \cite{GoulartKerrigan2006, KollerBerkenkampEtAl2018}.
In addition, we assume that for any $\fu \in \calU$ and $\hat\x \in \calX$, there exists a $\pi\in \calP_N$ such that $\fu = \pi(\x)$. We can trivially satisfy this assumption by, e.g., optimizing over open-loop nominal input sequences. Note that the estimator dynamics satisfy 
\begin{equation}\label{eq:est-dyn}
    \hat{\x}_{t+1} = \f(\hat{\x}_t + \pe_t, \fu_t, \w_t) - \pe_{t+1}.
\end{equation}
We bound the evolution of the state estimates over time for a given fallback policy sequence as follows:

\begin{restatable}[Reachable Sets]{lemma}{reachsets}\label{lem:reachable}
    Assume we apply a fixed fallback policy sequence $\fu_{0:T}^F \subset \calP_F$ from timestep $t$ to $t+T$. Define the $k-$step reachable sets of the estimate $\hat{\x}_t$ recursively as $\widehat{\calR}_{0}(\hat{\x}_t, \fu_{0:T}^F):= \{\hat{\x}_t\}$ and 
    \begin{align*}
        &\widehat{\calR}_{k+1}(\hat{\x}_t, \fu_{0:T}^F) := \bigg\{\substack{\f(\hat{\x} + \pe, \fu_{k}^F(\g(\hat{\x} + \pe)), \w) \\ - \pe'} \ : \substack{\hat{\x} \in \widehat{\calR}_{k}(\hat{\x}_t, \fu_{0:T}^F),\\ \w \in \calW, \\ \pe, \pe' \in \calE_\theta }\bigg\}
    \end{align*}
    for $k\in \{0, \dots, T\}$. Furthermore, let 
    \begin{equation}
        \calR_{k}(\hat{\x}_t, \fu^F_{0:T}) := \widehat{\calR}_{k}(\hat{\x}_t, \fu^F_{0:T}) \oplus \calE_{\theta}
    \end{equation}
    be the $k-$step reachable set of the true state $\x_t$. If $\pe_{t:t+T+1} \subset \calE_\theta$, then it holds that $\hat{\x}_{t+k} \in \widehat{\calR}_{k}(\hat{\x}_{t}, \fu_{0:T}^F) \subseteq \calR_{k}(\hat{\x}_{t}, \fu_{0:T}^F)$ and $\x_{t+k} \in \calR_{k}(\hat{\x}_{t}, \fu_{0:T}^F)$ for all $k \in \{0, \dots, T+1\}$. Moreover, it holds that $\calR_{k}(\hat{\x}_{t+1}, \fu^F_{1:T}) \subseteq \calR_{k+1}(\hat{\x}_{t}, \fu^F_{0:T})$ for $k\in \{0, \dots, T \}$.
\end{restatable}
\begin{proof}
    % See \cite{SinhaSchmerlingEtAl2023} for the proof.
    See \cref{ap:reachable} for the proof.
\end{proof}

To maintain feasibility of the fallback policy despite estimation errors and disturbances, we solve the following finite time robust optimal control problem online:

% \begin{equation}\label{eq:modification-mpc}
% \begin{aligned}
%     \minimize_{\substack{\fu^R_{t:t+T|t}, \\ \fu_{t:t+T|t}}}
%     & \enspace C(\hat{\x}_{t|t}, \fu_{t:t+T|t}, \fu^R_{t:t+T|t})\\\
%     \subjectto\enspace & \x_{t+k+1|t} = \f(\x_{t+k|t}, \fu^R_{t+k|t}, \w_{t+k|t}) + \pe_{t+k+1|t},\\
%     & \x_{t+k|t} \in \calX, \ \{\fu^R_{t+k|t}, \ \fu_{t+k|t}\} \subset \calU, \ \forall k \in \{0, \dots, T\}, \\
%     & \x_{t+T+1|t} \in \calX_R, \quad \x_{t|t} = \hat{\x}_t + \pe_{t|t}, \\
%     & \fu_{t|t} = \fu^R_{t|t}, \\
%     & \forall \{\w_{t+k|t}\}_{k=0}^{T} \subset \calW, \quad \forall \{\pe_{t+k|t}\}_{k=0}^{t+T+1} \subset \calE_\theta. \\
% \end{aligned}
% \end{equation}

\begin{equation}\label{eq:modification-mpc}
\begin{aligned}
    \minimize_{\substack{\fu^F_{t:t+T|t}\subset \calP_F, \\ \fu_{t:t+T|t} \subset \calP_N}} 
    & \enspace C(\hat{\x}_{t}, \fu_{t:t+T|t}, \fu^F_{t:t+T|t})\\
    \subjectto\enspace & \calR_{k}(\hat{\x}_t, \fu_{t:t+T|t}^F) \subseteq \calX  \ \forall k \in \{0, \dots, T\}, \\
    &\fu^F_{t+k|t}(\g(\calR_{k}(\hat{\x}_t, \fu_{t:t+T|t}^F))) \subset \calU \ \forall k \in \{0, \dots, T\}, \\
    & \calR_{T+1}(\hat{\x}_t, \fu^F_{t:t+T|t}) \subseteq \calX_R, \\
    & \fu_{t|t}(\hat{\x}_t) = \fu^F_{t|t}(\y_t).
\end{aligned}
\end{equation}
The MPC problem \cref{eq:modification-mpc} robustly optimizes the trajectory of the robot along a $T+1$ step prediction horizon and maintains both a nominal policy sequence $\fu_{t:t+T|t}$, and a fallback tube $\calR_{0:{t+T+1}}(\hat{\x}_t, \fu_{t:t+T|t}^F)$. Let $\{\fu_{t+k|t}^\star, \fu^{F, \star}_{t+k|t}\}_{k=0}^T$ be an optimal collection of policy sequences for \eqref{eq:modification-mpc} at time step $t$. 
Executing the fallback policy sequence $\fu_{t:t+T|t}^{F,\star}$ guarantees that we reach a given recovery set $\calX_R$ within $T+1$ time steps for any disturbance sequence $\w_{t:t+T|t}$ and perception errors $\pe_{t:t+T+1|t} \subset \calE_\theta$. 
% In addition, we account for the state uncertainty -- that is, that we observe $\hat{\x}_t = \x_t -\pe_t$ instead of $\x_t$ at each time step -- by propagating the estimation errors $\pe_{t:t+T+1|t}$ through the dynamics in the construction of the reachable sets. 
Because we ensure that the first inputs of both the nominal and the fallback policies are identical, i.e., that $\fu_{t|t}^F(\y_t) = \fu_{t|t}(\hat{\x}_t)$, we can guarantee that we can recover the system to $\calX_R$ if we detect a fault at $t+1$ by applying the current fallback policy sequence $\fu_{t:t+T|t}^{F,\star}$. 

We assume the recovery set is invariant with respect to the estimator dynamics \cref{eq:est-dyn} in nominal conditions.
\begin{assumption}\label{as:recovery}
    We are given a \emph{recovery policy} $\pi_R: \calY \to \calU$ associated with a nonempty recovery set $\calX_R$ under the estimate dynamics \cref{eq:est-dyn} in nominal conditions, so that $\calR_1(\hat{\x}, \pi_R(\g(\cdot))) \subseteq \calX_R$ for all $\hat{\x}\in\calX_R$. 
\end{assumption}
% Similar to the standard robust MPC assumption, that we have a terminal controller associated with a nonempty RPI set \cite{BorrelliBemporadEtAl2017, MayneRakovic2006, GoulartKerrigan2006},  
\cref{as:recovery} follows the classic assumption in the robust MPC literature---access to a terminal controller associated with a nonempty RPI set---that enables guarantees on persistent feasibility and constraint satisfaction \cite{BorrelliBemporadEtAl2017, MayneRakovic2006, GoulartKerrigan2006, LorenzettiPavone2020}, except that we explicitly enforce that the recovery policy only uses fallback measurements $\y_t$. We note that \cref{as:recovery} can be satisfied by designing a recovery policy (e.g., LQR or human-insight as in the drone landing example) and verifying whether a chosen set satisfies \cref{def:recovery} offline. Alternatively, we can compute $\calX_R$ using existing algorithms for robust invariant set computation (e.g., see \cite{BorrelliBemporadEtAl2017}), since under $\pi_R$ and the assumption that $\pe\in\calE_\theta$, \cref{eq:est-dyn} is an autonomous system subject to bounded disturbances. 

% The objective in problem \cref{eq:modification-mpc} is to minimize a chosen cost function $C$.
We choose the objective $C$ in problem \cref{eq:modification-mpc} to minimally interfere with the nominal operation of the robot by optimizing a disturbance free nominal trajectory as is common in the MPC literature (e.g., see \cite{BorrelliBemporadEtAl2017, GoulartKerrigan2006, MayneRakovic2006,BrunkeGreefEtAl2022}). An alternative is to minimally modify the outputs of another controller \cite{BrunkeZhou2022, ChengOrosz2019}. 
% That is, we could select an objective of the form
% \begin{align}
%     &C(\hat{\x}_{t}, \fu_{t:t+T|t}, \fu^F_{t:t+T|t}) = C_{\mathrm{nom}}(\hat{\x}_{t}, \fu_{t:t+T|t}) \nonumber \\
%     &:= \E\bigg[ \sum_{k=0}^{T} h(\x_{t+k}, \fu_{t+k|t}) + V(\x_{t+T+1}) \ \bigg|\ \hat{\x}_{t}, \fu_{t:t+N|t} \bigg] , \label{eq:nomcost}
% \end{align}
% for a suitably chosen stage cost function $h$ and terminal cost function $V$ and i.i.d. (and perhaps zero mean) stochastic disturbance models for $\w_t$ and $\pe_t$. 
% However, we recognize that we may want to accept some suboptimality in the nominal regime so that the fallback trajectories are not highly suboptimal. Therefore, we use objectives of the form 
% \begin{equation}
%     C(\cdot) = C_{\mathrm{nom}}(\hat{\x}_{t}, \fu_{t:t+N|t}) + \lambda C_{\mathrm{rec}}(\hat{\x}_{t}, \fu^F_{t:t+N|t})
% \end{equation}
% in our simulations.
% % Here, $C_{\mathrm{rec}}$ is an objective on the expected recovery trajectory similar to \cref{eq:nomcost} and 
% Here, $\lambda \geq 0$ is a hyperparameter to make a trade-of between optimizing nominal operations and having a smooth fallback manoeuvre. For example, if we can estimate the probability $p$ of triggering the fallback sometime in the future, setting $\lambda = p / (1 - p)$ will optimize the mean cost of the system, similar to the contingency MPC in \todo{cite gerdes}.  
% We note that the robust MPC problem in \cref{eq:modification-mpc} is intractable as written. take advantage of the robust MPC literature to construct
We develop our framework in generic terms in this section, so we provide a tractable reformulation of \cref{eq:modification-mpc} for linear-quadratic systems with a fixed feedback gain based on classic tube MPC algorithms \cite{MayneRakovic2006} in 
% \cref{ap:mpc-linear}.
\cite{SinhaSchmerlingEtAl2023}.
For nonlinear systems, it is common to approximate a solution via sampling (e.g., \cite{DyroHarrisonEtAl2021, LewJansonEtAl2022}), so we also provide an approximate formulation using the PMPC algorithm \cite{DyroHarrisonEtAl2021} in 
% \cref{ap:mpc-particle},
\cite{SinhaSchmerlingEtAl2023},
which combines uncertainty sampling with sequential convex programming.

 Here we assume access to a runtime monitor $w$ to decide when we trigger the fallback; we construct a monitor with provable guarantees in \cref{sec:conformal}.
% We require a runtime monitor to decide when we trigger the fallback. 
\begin{definition}[Runtime Monitor]
    A runtime monitor $w: \R \to \{0,1\}$ is a function of the anomaly signal, where $w(a_t) = 1$ implies the monitor raises an alarm.
    % In addition, let $t_{\mathrm{fail}}(t) := \inf_{t'\leq t} \{t' \ : \ w(a_t') = 1\}$ be the first time step at which the runtime monitor raises an alarm.
    % so that $w(a_t) = 1$ implies that a perception fault was detected
\end{definition}
We solve \cref{eq:modification-mpc} online at each time step $t$ and apply the first optimal control input in a receding horizon fashion. If the runtime monitor triggers, indicating a detection of a perception fault (i.e., $\x_t - \hat{\x}_t \not\in \calE_t$), we apply the previously computed fallback policy sequence until we reach the recovery set. Then, we revert to the known recovery policy. We summarize this procedure in Algorithm \ref{alg:fallback-safe-mpc}. 

Let $t_{\mathrm{fail}}$ be the timestep at which the runtime monitor $w$ triggers the fallback. As long as the runtime monitor does not miss a detection of a perception fault, i.e., that $\x_t - \hat{\x}_t \in \calE_\theta$ for all $0 \leq t < t_{\mathrm{fail}}$, the MPC in \cref{eq:modification-mpc} is recursively feasible. As a result, Algorithm \ref{alg:fallback-safe-mpc} persistently satisfies state and input constraints in the presence of disturbances and estimation errors.
% Next, we certify that our Fallback-Safe MPC scheme ensures the safety of the system for all time, as long as the runtime monitor triggers the fallback before the estimation errors violate our specified perception error bound.
% The policy \cref{eq:nominal-policy} first uses \cref{eq:modification-mpc} to modify the nominal behavior of the system. Then when the runtime monitor triggers, it applies the previously computed fallback input sequence in open-loop until the recovery set is reached, after which we revert to the known recovery policy. We now prove that this MPC scheme ensures the feasibility of the fallback strategy and thereby guarantees the system safety, as long as our perception errors follow the specified bound $\calE_\theta$.

% \begin{equation}
    % \x_t - \hat{\x}_t \in \calE_\theta \quad \forall \ 0 \leq t < t_{\mathrm{fail}}.
% \end{equation}
% We then apply the input
% \begin{equation}\label{eq:nominal-policy}
%     \fu_t = \pi(\hat{\x}_t, \y_t) := \begin{cases}
%     \fu^\star_{t|t}(\hat{\x}_t) & \mathrm{if} \quad t_{\mathrm{fail}}(t) > t  \\
%     \fu^{F,\star}_{t|t_{\mathrm{fail}} - 1}(\y_t) & \mathrm{if} \quad t_{\mathrm{fail}}(t) \leq t < t_{\mathrm{fail}}(t) + T \\
%     \pi_R(\y_t) & \mathrm{if} \quad t \geq t_{\mathrm{fail}}(t) + T
%     \end{cases} 
% \end{equation}

\begin{algorithm}[t]\label{alg:fallback-safe-mpc}
\caption{Fallback-Safe MPC}
  \SetAlgoLined
  \KwIn{Initial state estimate $\hat\x_0$ such that \cref{eq:modification-mpc} is feasible, 
  runtime monitor $w:\R \to \{0,1\}$.}
  $t_{\mathrm{fail}} \gets \infty$

  \For{ $t \in \{0, 1, \dots t_{\mathrm{lim}}\}$}{
    Observe $\hat{\x}_t, \ \y_t,\ a_t$
    
    \If{$w(a_t) = 1$}{ % or \cref{eq:modification-mpc} infeasible at $\hat{\x}_t$}{
    $t_{\mathrm{fail}} \gets \min\{t_{\mathrm{fail}}, \ t\}$}
    Apply control input 
    \begin{equation*}
        \fu_t =\begin{cases}
    \fu^\star_{t|t}(\hat{\x}_t) & \mathrm{if} \quad t_{\mathrm{fail}} > t  \\
    \fu^{F,\star}_{t|t_{\mathrm{fail}} - 1}(\y_t) & \mathrm{if} \quad t_{\mathrm{fail}} \leq t < t_{\mathrm{fail}} + T \\
    \pi_R(\y_t) & \mathrm{if} \quad t \geq t_{\mathrm{fail}} + T
    \end{cases} 
    \end{equation*}
    }
\end{algorithm}
% at time $t$. 
\begin{restatable}[Fallback Safety]{theorem}{fallbacksafe}\label{thm:mpc-safety}
Consider the closed-loop system formed by the dynamics \cref{eq:dynamics} and the Fallback-Safe MPC (Algorithm \ref{alg:fallback-safe-mpc}). Suppose that $\pi_R \in \calP_F$, and 
% that the runtime monitor $w$ does not miss a detection of a perception fault, i.e., 
that $\x_t - \hat{\x}_t \in \calE_\theta$ for all $0 \leq t < t_{\mathrm{fail}}$. Then, if the Fallback-Safe MPC problem \cref{eq:modification-mpc} is feasible at $t = 0$ and $w(a_0) = 0$, we have that 1) the MPC problem \cref{eq:modification-mpc} is feasible for all $t < t_{\mathrm{fail}}$ and that 2) the closed-loop system  satisfies $\x_t \in \calX$, $\fu_t \in \calU$ for all $t \geq0$.
\end{restatable}
\begin{proof}
% See \cite{SinhaSchmerlingEtAl2023} for the proof.
See \cref{ap:mpc-safety} for the proof.
\end{proof}

We emphasize that \cref{thm:mpc-safety} simply recovers a standard recursive feasibility argument for the MPC scheme in the specific case in which a perception failure never occurs.  % \begin{corollary}
\section{CALIBRATING OOD DETECTORS WITH CONFORMAL INFERENCE}\label{sec:conformal}

In the previous section, we developed the Fallback-Safe MPC framework, which guarantees safety under the condition that the perception is reliable at all time steps before we trigger the fallback.  We could trivially ensure this is the case by setting $w(a) = 1 \ \forall a \in \R \setminus a_0$, so that the fallback always triggers, no matter the quality of the perception. However, a trivial runtime monitor will unnecessarily disrupt nominal operations, so it is not useful. Therefore, in this section, we aim to construct a runtime monitor $w:\R \to \{0,1\}$ that provably triggers with high probability when a perception fault occurs, but does not raise too many false alarms in practice. To do so, we adapt the conformal prediction algorithm in \cite{LuoZhao2023}, which can only certify a prediction on a single test point, to retain a safety assurance when we sequentially query the runtime monitor online with the anomaly scores $a_0, a_1, \dots$ generated during a test trajectory. Our procedure requires some ground truth data to calibrate the runtime monitor. As a shorthand, we use the notation $\tau$ to denote a trajectory with ground truth information, $\tau := (\{\x_i, \fu_i, \hat{\x}_i, a_i\})_{i=0}^{t_{\mathrm{lim}}} \in \calT$.

% %to say:
% % 1. that the previous MPC requires us to trigger the fallback before or exactly when the perception becomes inaccurate.
% % 2. therefore, we construct a runtime monitor that guarantees this fact is true with probability at least 1 - delta.
% % 3. we modify rachel's conformal algorithm to suit this purpose
% % 5. because conformal inference is much more sample efficient than a traditional PAC style learning framework that relies on concentration inequalities and uniform convergence.
% % 6. i
% Step 1: Use the notation $\tau$ to denote a ground truth trajectory (of potentially infinite length), $\tau := (\x_1, \fu_1, \x_2, \fu_2, \dots) \in \calT$.

% Step 2: Assume we have a validation testing dataset of trajectories.
\begin{assumption}[Calibration Data]
    We have access to a trajectory dataset $\calD = \{\tau_i\}_{i=0}^N \iid P$ sampled independently and identically distributed (iid) from a trajectory distribution $P$.
\end{assumption}

% However, \cref{lem:conformal} shows that the coverage guarantee degrades when a distribution shift occurs between the calibration runs in $\calD$ and the test trial. 
The trajectory distribution $P$ is a result of 1) the environment distribution $P_\rho$ and 2) the controller we use to collect data. Therefore, in general, when we deploy the Fallback-Safe MPC policy (Alg. \ref{alg:fallback-safe-mpc}), the resulting trajectory distribution $P'$ may differ from $P$. The results we present in this section capture both the scenario where the environment distribution changes between data collection and deployment or where the data collection policy differs from the Fallback-Safe MPC.

% Such a shift can be caused by an external factor, represented as a change to the context distribution $P_\rho$. Therefore, in general, it is not possible to tightly bound the $\mathrm{TV}$ distance term in \cref{eq:coverage}, so safety assurances can only be made by assuming invariance of the environment generating distribution. Moreover, a shift between $P$ and $P'$ can occur because the Fallback-Safe MPC policy (Alg. \ref{alg:fallback-safe-mpc}) that we deploy at test time can differ from the policy used to collect $\calD$.

% We make a coverage guarantee capturing both aforementioned scenarios in \cref{lem:conformal}, but we will apply our Fallback-Safe MPC in a context where the environment distribution is fixed and we choose the controller used for data collection to enforce that $P_{t_{\mathrm{stop}}|\mathrm{fault}}=P'_{t_{\mathrm{stop}}|\mathrm{fault}}$ so that we can certifiably satisfy \cref{def:safety} in our simulations. 

\cref{thm:mpc-safety} requires that we trigger the fallback at any time step before or when a perception fault occurs. Therefore, we must compare the step that the runtime monitor triggers an alarm, $t_{\mathrm{fail}}$, with the first step that a perception fault occurs.
\begin{definition}[Stopping Time]
    Let $t_{\mathrm{stop}} : \calT \to \mathbb{N}_{\geq 0}$ be the stopping time of a trajectory $\tau$, defined as
    \begin{equation}
        t_{\mathrm{stop}}(\tau) := \inf_{t\geq0} \{t \ : \ (\x_t -\hat{\x}_t \not \in \calE_\theta) \ \vee \ (t = t_{\mathrm{lim}})\},
    \end{equation}
    where $t_{\mathrm{lim}}$ is the episode time limit in \cref{def:safety}. %, so that $t_{\mathrm{stop}} \leq t_{\mathrm{lim}}$.
\end{definition}
To guarantee safety as in \cref{def:safety}, our insight is that it is sufficient to guarantee that our runtime monitor raises an alarm at the \emph{first} time step for which $\pe\not \in \calE_{\theta}$ with high probability. Therefore, we can circumvent the need for a complex sequential analysis that accounts for the correlations between the runtime monitor's hypothesis tests over time. Instead, we can directly apply methods developed for i.i.d. samples to the dataset of stopping time observables $\calD_{\mathrm{stop}}:= \{(\x^i_{t_{\mathrm{stop}}(\tau_i)}, \hat\x^i_{t_{\mathrm{stop}}(\tau_i)}, a^i_{t_{\mathrm{stop}}(\tau_i)})\}_{i=0}^N$, since the stopping time observables are i.i.d. because $\calD$ is i.i.d. We outline this approach in Algorithm \ref{alg:conformal}, which adapts the conformal prediction algorithm in \cite{LuoZhao2023} to our sequential setting.

\begin{algorithm}[t]\label{alg:conformal}
\caption{Modification of \cite{LuoZhao2023} for Conformal Calibration of Runtime Monitor}
  \SetAlgoLined
  \KwIn{Dataset $\calD = \{\tau_i\}_{i=0}^N \iid P$, perception system, OOD detector, state uncertainty tolerance $\calE_\theta$, risk tolerance $\delta \in (0,1]$, new test anomaly score $a_{\mathrm{test}} \in \R$.}
  \KwOut{0 or 1}
  Compute the dataset of stopping states, estimates, and anomaly scores as
\begin{equation*}
    \calD_{\mathrm{stop}}:= \{(\x^i_{t_{\mathrm{stop}}(\tau_i)}, \ \hat\x^i_{t_{\mathrm{stop}}(\tau_i)}, \ a^i_{t_{\mathrm{stop}}(\tau_i)})\}_{i=0}^N.
\end{equation*}

  Compute the set 
\begin{equation*}
    \calA := \{a \ : \ \x - \hat{\x} \not \in \calE_{\theta}, \ (\x, \hat{\x}, a) \in \calD_{\mathrm{stop}}\}.
\end{equation*}

  Sample $U$ uniformly from
\begin{equation*}
      U \sim \{0,1,\dots, |\{a \in \calA: a = a_{\mathrm{test}}\}|\}
\end{equation*}

Compute 
\begin{equation*}
    q = \frac{|\{a \in \calA : a > a_{\mathrm{test}}\}| + U + 1}{|\calA| + 1}
\end{equation*}
    % Compute the scaled $1 - \delta$ quantile
    % \begin{equation*}
    %     \hat{q} := \inf_{a} \{a \ :\ |\{a_i \in \calA \ : \ a_i \leq a\}| \geq \ceil{(|\calA|+1)(1 - \delta)}\}.
    % \end{equation*}
  % \eIf{$a \geq \hat{q}$}{ %_{a|\mathrm{fail}}(1 - \delta)
  %   \KwRet{$1$}
  % }{
  %   \KwRet{$0$}
  % }

  \textbf{if} $q \leq 1 - \delta$ \textbf{then return} 1 \textbf{else return} $0$
\end{algorithm}
% Step 3: Define the stopping time $t_{\mathrm{stop}}(\tau) = \min_{t\geq0} \{t \ : \ (\x_t -\hat{x}_t \not \in \calE_\theta) \ \vee \ (\x_t \in \calX_g)\}$, where $\calX_g$ is some goal region so that $t_{\mathrm{stop}} < \infty$ almost surely.

% Step 4: Take the calibration dataset $\{\x^i_{t_{\mathrm{stop}}}(\tau_i)\}_{i=0}^M$

% Step 5: Apply Rachel's algorithm from \cite{LuoZhaoEtAl2021} on the calibration dataset to detect $\x_t -\hat{x}_t \not \in \calE_\theta$ with a provable false negative rate by thresholding $a_t$.

% Step 6: Wrap with the MPC to make an end-to-end guarantee.

\begin{restatable}[Conformal Calibration]{lemma}{conformal}\label{lem:conformal}
    Set a risk tolerance $\delta \in (0,1]$, and sample a deployment trajectory $\tau\sim P'$ by executing the Fallback-Safe MPC (Algorithm \ref{alg:fallback-safe-mpc}) using Algorithm \ref{alg:conformal} as runtime monitor $w$.
    % Assume the environment $\rho$ is independent of $\calD$. Let $P'$ denote the distribution of $\tau$.
    % Furthermore, let $(a_{t_{\mathrm{stop}}(\tau)}, \x_{t_\mathrm{stop}(\tau)})$ be the anomaly score and true state at the stopping time of trajectory $\tau$.
    % and let $w(a)$ denote the output of Algorithm~\ref{alg:conformal} on a test input $a\in\R$. 
    Then, the false negative rate of Algorithm~\ref{alg:conformal} is bounded as 
    \begin{equation}\label{eq:coverage}
    \begin{aligned}
        &\prob(\mathrm{False \ Negative}) := \\
        &\prob\big(
        % w(a_{t_{\mathrm{stop}}(\tau)}) = 0
        w(a_t) = 0 \ \forall t \in [0, t_{\mathrm{stop}}(\tau)]
        \ | \ \pe_{t_{\mathrm{stop}}(\tau)} \not \in \calE_{\theta} \big) \leq \\
        &\quad \delta + \frac{1}{|\calA| + 1} + \mathrm{TV}(P_{t_{\mathrm{stop}}|\mathrm{fault}}, P'_{t_{\mathrm{stop}}|\mathrm{fault}}),
    \end{aligned} 
    \end{equation}
    where $P_{t_{\mathrm{stop}}|\mathrm{fault}}$ is the distribution of $(\x, \hat{\x}, a)$ at $t_{\mathrm{stop}}$ conditioned on the event that $\pe_{t_{\mathrm{stop}}} \not \in \calE_\theta$ under a trajectory sampled from $P$, and $P_{t_{\mathrm{stop}}|\mathrm{fault}}'$ is the distribution of $(\x, \hat{\x}, a)$ at $t_{\mathrm{stop}}$ under a trajectory sampled from $P'$,  conditioned on both $\pe_{t_{\mathrm{stop}}} \not \in \calE_\theta$ and  $t_{\mathrm{fail}}\geq t_{\mathrm{stop}}$. Here, $TV(\cdot, \cdot)$ denotes the total variation distance.   
\end{restatable}
\begin{proof}
    % See \cite{SinhaSchmerlingEtAl2023} for the proof.
    See \cref{ap:conformal} for the proof. 
\end{proof}
\cref{lem:conformal} shows that Algorithm \ref{alg:conformal} will issue a timely warning with probability at least $\delta + 1 / (|\calA| + 1)$ without relying on properties of $P$ (i.e., our guarantee is distribution-free), but that this guarantee degrades when a distribution shift occurs between the calibration runs in $\calD$ and the test trial. 
Next, we leverage \cref{lem:conformal} to analyze the end-to-end safety of the system.
% 1. very rare on failure modes
% 2. much less than what's needed to retrain
% 3. run pre-deployment trials of a pretrained system in a new context
% 4. statement about pac-learning.

% to say
% 1. smt about the degradation with respect to dist shift
% 2. smt about how we can make claims on te false positive rate by analyzing $max a_t$ st t<t_stop
% 3. something about the sample complexity wrt pac learning
% 4. that we use this result to make the end-to-end-guarantee
% 5. something about the quality of the heuristic

\begin{restatable}[End-to-end Guarantee]{theorem}{endtoend}\label{thm:endtoend}
    Consider the closed-loop system formed by the dynamics \cref{eq:dynamics} and the Fallback Safe MPC (Algorithm \ref{alg:fallback-safe-mpc}), using Algorithm \ref{alg:conformal} as the runtime monitor $w$. Then, if the MPC problem \cref{eq:modification-mpc} is feasible at $t=0$ and $w(a_0) = 0$, it holds that
    \begin{align*}
        &\prob(\x_t \in \calX, \ \fu_t \in \calU \ \forall t \in [0, t_{\mathrm{lim}}]) \geq \\
        &1 - \delta - \frac{1}{|\calA| + 1} - \mathrm{TV}(P_{t_{\mathrm{stop}}|\mathrm{fault}}, P_{t_{\mathrm{stop}}|\mathrm{fault}}').
    \end{align*}
\end{restatable}
\begin{proof}
    % For the proof, see \cite{SinhaSchmerlingEtAl2023}.
See \cref{ap:endtoend} for the proof.
\end{proof}

\cref{thm:endtoend} gives a general end-to-end guarantee on the safety of the Fallback-Safe MPC framework when we use Algorithm \ref{alg:conformal} as a runtime monitor. It is not possible to tightly bound the $\mathrm{TV}$ distance term in \cref{eq:coverage} without further assumptions.
% so safety assurances can only be made by assuming invariance of the environment generating distribution. Moreover, a shift between $P$ and $P'$ can occur because the Fallback-Safe MPC policy (Alg. \ref{alg:fallback-safe-mpc}) that we deploy at test time can differ from the policy used to collect $\calD$. 
% We make a coverage guarantee capturing both aforementioned scenarios in \cref{lem:conformal},
However, if  1) we use the Fallback-Safe MPC for data collection and 2) the environment distribution is fixed between data collection and deployment, then we can certify that we satisfy \cref{def:safety}:

% apply our Fallback-Safe MPC in a context where the environment distribution is fixed and we choose the controller used for data collection to enforce that $P_{t_{\mathrm{stop}}|\mathrm{fault}}=P'_{t_{\mathrm{stop}}|\mathrm{fault}}$ so that we can certifiably satisfy \cref{def:safety} in our simulations. 

% To satisfy \cref{def:safety} under the condition that $P_{t_{\mathrm{stop}}|\mathrm{fault}}=P'_{t_{\mathrm{stop}}|\mathrm{fault}}$, we note that \ref{lem:conformal} shows that Algorithm \ref{alg:conformal} will detect the first perception fault with probability at least $\delta + 1 / (|\calA| + 1)$, without relying on properties of $P$ (i.e., our guarantee is distribution-free). We can then satisfy \cref{def:safety} for a risk tolerance $\delta$ by evaluating Algorithm \ref{alg:conformal} using $\delta' = \delta - 1 / (|\calA| + 1)$ in combination with the Fallback-Safe MPC. As long as we have sufficient data on failure modes, that is, when $(1 / \delta) - 1 \leq |\calA|$, our runtime monitor will exhibit nontrivial behavior (i.e., that $w$ does not always output $1$). 

% Next, we have the following corollary that outlines how we can certify the end-to-end safety of the Fallback-Safe MPC:
\begin{corollary}\label{cor:endtoend}
Suppose the environment distribution $P_\rho$ is fixed between collecting $\calD$ and the test trajectory, and that we collect the dataset $\calD$ by running the Fallback-Safe MPC and a runtime monitor using privileged information, that is, $w(\cdot) := 1 - \mathbf{1}\{\pe_t \in \calE_\theta\}$. Then, it holds that $P_{t_{\mathrm{stop}}|\mathrm{fault}} = P'_{t_{\mathrm{stop}}|\mathrm{fault}}$. Therefore, 1) we satisfy state and input constraints during data collection with probability $1$, and 2) we satisfy \cref{def:safety} with probability at least $1 - \delta - \frac{1}{|\calA| +1}$ during a test trajectory. 
\end{corollary}
% \begin{proof}
%     If $P_\rho$ is constant and we collect data with Algorithm \ref{alg:fallback-safe-mpc} and a ground-truth monitor, then $P_{t_{\mathrm{stop}}|\mathrm{fault}} = P'_{t_{\mathrm{stop}}|\mathrm{fault}}$. The corollary then follows from \cref{thm:endtoend}. 
% \end{proof}

\cref{cor:endtoend} informs the following two-step procedure to yield a provable end-to-end safety guarantee on a fixed environment distribution $P_\rho$. We use this procedure in \cref{sec:simulations}. First, collect $\calD$ using the Fallback-Safe MPC with a ground-truth supervisor $w(\cdot) := 1 - \mathbf{1}\{\pe_t \in \calE_\theta\}$, then deploy the Fallback-Safe MPC with Algorithm \ref{alg:conformal} as the runtime monitor. 
% To satisfy \cref{def:safety} under the condition that $P_{t_{\mathrm{stop}}|\mathrm{fault}}=P'_{t_{\mathrm{stop}}|\mathrm{fault}}$, we note that \ref{lem:conformal} shows that Algorithm \ref{alg:conformal} will detect the first
We can then satisfy \cref{def:safety} for a risk tolerance $\delta$ by evaluating Algorithm \ref{alg:conformal} using $\delta' = \delta - 1 / (|\calA| + 1)$. As long as we have sufficient data on failure modes, that is, when $(1 / \delta) - 1 \leq |\calA|$, our runtime monitor will exhibit nontrivial behavior (i.e., that $w$ does not always output $1$). 
% We use this procedure in \cref{sec:simulations}.

% \begin{remark}
%     \todo{cut this?} We can actually significantly strengthen our result in \cref{thm:endtoend} because we make a guarantee on the false negative rate in \cref{lem:conformal}, which is a conditional probability: The false negative rate guarantee holds for all distributions over trajectories conditioned on the fact that a perception fault occurs. This means that as long as this conditional distribution stays the same, our end-to-end safety guarantee in \cref{thm:endtoend} regardless of 1) the distribution of trajectories without failures, and 2) the prior likelihood of experiencing a failure at test time. I.e., we achieve safety with probability at least $1 - \delta -  1 / (|\calA| + 1)$ for all test distributions $P' = (1 - \alpha) P'_{\mathrm{no fail}} + \alpha P_{\mathrm{fail}}$ with $\alpha \in [0,1]$. 
% \end{remark}

\section{SIMULATIONS}\label{sec:simulations}
In this section, we first simulate a simplified example of a quadrotor to illustrate the behavior of the Fallback-Safe MPC framework. We then demonstrate the efficacy of the conformal algorithm, and the resulting end-to-end safety guarantee, in the photo-realistic X-Plane 11 aircraft simulator. For a detailed description of our simulations, including e.g., the specific cost functions used, we refer the reader to 
% \cref{ap:sim}.
\cite{SinhaSchmerlingEtAl2023}.

\textbf{Planar Quadrotor:} We consider a planar version of the quadrotor dynamics for simplicity, with 2D pose $\bm{p} = [x, y, \theta]^T$, state $\x = [\bm{p}^T, \dot{\bm{p}}^T]^T$, and front and rear input thrust inputs $\fu = [u_f, u_r]^T$ \cite{Tedrake2021}. We linearize the the dynamics around $\bar{\x} = 0$ and $\bar{\fu} = \frac{mg}{2}[1,1]^T$, and discretize the dynamics using Euler's method with a time step of $\mathrm{dt} =0.15s$. The drone is subject to bounded wind disturbances, so that the drone may drift with $\approx 0.33 m / s$ in the $x-$direction without actuation. In our example, the drone has internal sensors to estimate its orientation and velocity, so that $\y = [\theta, \dot{\bm{p}}^T]^T$. The drone estimates its $xy-$position using a hypothetical vision sensor. To do so, we nominally simulate that the perception system's $xy-$position estimate is within a $10 \mathrm{cm}$-wide box around the true position, and perfectly outputs $\y$. When the vision system fails, we randomly sample the $xy-$position within the range $(-10,10)$. In these simulations we give the Fallback-Safe MPC a perfect runtime monitor, so that $w(\cdot) = 1 - \mathbf{1}\{\pe_t \in \calE\}$. We implement the Fallback-Safe MPC using the tube MPC formulation in 
% \cref{ap:mpc-linear}
\cite{SinhaSchmerlingEtAl2023}.

First, in \cref{fig:landing}, we simulate a scenario where the drone attempts a vision-based landing at the origin. Here, the state constraint is not to crash into the ground ($y\leq0$). We set the recovery policy to $\pi_R(\y):= \epsilon + K \y$, where $K$ stabilizes the orientation $\theta$ around $0$ and $\epsilon$ is a small offset to continually fly upward, choosing the recovery set to allow the drone to fly away starting from a sufficient altitude.
We verify that under the recovery policy, the recovery set is invariant under both the estimator dynamics \cref{eq:est-dyn} in nominal conditions and the state dynamics \cref{eq:dynamics}, so the Fallback-Safe MPC is recursively feasible by \cref{thm:mpc-safety}. We compare the Fallback-Safe MPC with a 
% naive tube MPC that optimizes only a single trajectory and does not anticipate vision failures. 
naive tube MPC that optimizes
only a single trajectory and assumes perception is always reliable (i.e., it assumes perception errors always satisfy \cref{eq:error})
As shown in \cref{fig:landing}, the Fallback-Safe MPC plans fallback trajectories that safely abort the landing and fly away into open space. In contrast, the naive tube MPC does not reason about perception faults, and crashes badly when the perception fails starting at $t = 10\mathrm{dt}$. Therefore, this example demonstrates the necessity of planning with a fallback. 
% We also see that the perception uncertainty prevents the Fallback-Safe MPC from completing the landing to ensure the true state does not leave $\calX$.
\begin{figure}[t]
    \centering
    \includegraphics[width=\textwidth]{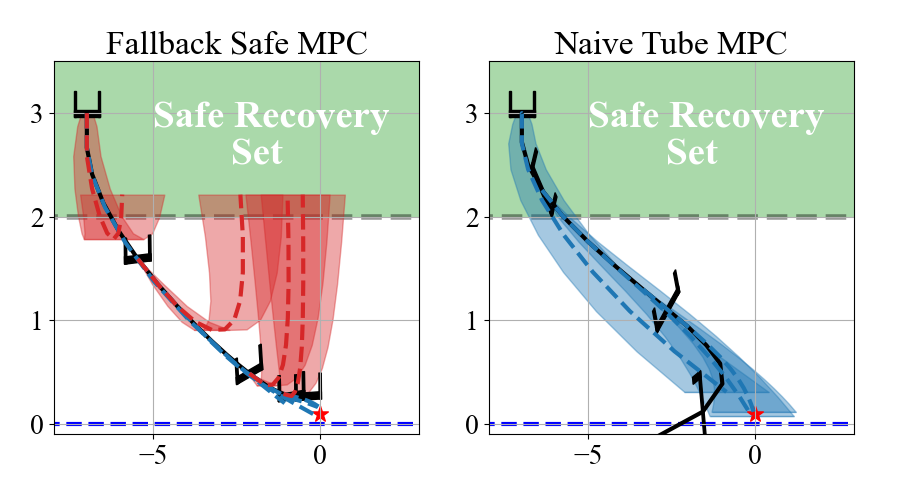}
    \caption{Trajectories of the planar quadrotor in the $xy$-plane. The realized quadrotor trajectories are in black, and the icons show the orientation of the quadrotor at every $k=7$ time steps. The safe recovery set is highlighted in green. The blue-dashed line indicate the state constraint. Left: In red, we plot the predicted reachable sets of the fallback strategy and in blue, we plot the predicted nominal trajectories, both at $k=7$ step intervals. Right: In blue, we plot the predicted reachable tubes of the Naive Tube MPC at $k$ step intervals.}
    \label{fig:landing}
    \vspace{-0.5cm}
\end{figure}

Secondly, we simulate a scenario where the drone must navigate towards an in-air $xy$ goal location while remaining within a box in the $xy$-plane. Here, when the drone loses its vision, it is no longer possible to avoid the boundaries of $\calX$ using only the fallback measurement $\y$. Instead, as in the example in \cref{fig:hero-fig}, our recovery set is to land the drone. To model the drone as having landed, we modify the dynamics to freeze the state for all remaining time once the state $\x$ enters the $y \leq 0$ region with low velocity. However, in this example, the drone must cross an unsafe ground region, such as the busy road in the example in \cref{fig:hero-fig}, specified as the region of states with $|x| < 1, \ y\leq 0$. Therefore, we take the recovery set $\calX_R$ as all landed states with $|x|\geq 1$. 
\begin{figure}[b]
    \centering
    \includegraphics[width=\textwidth]{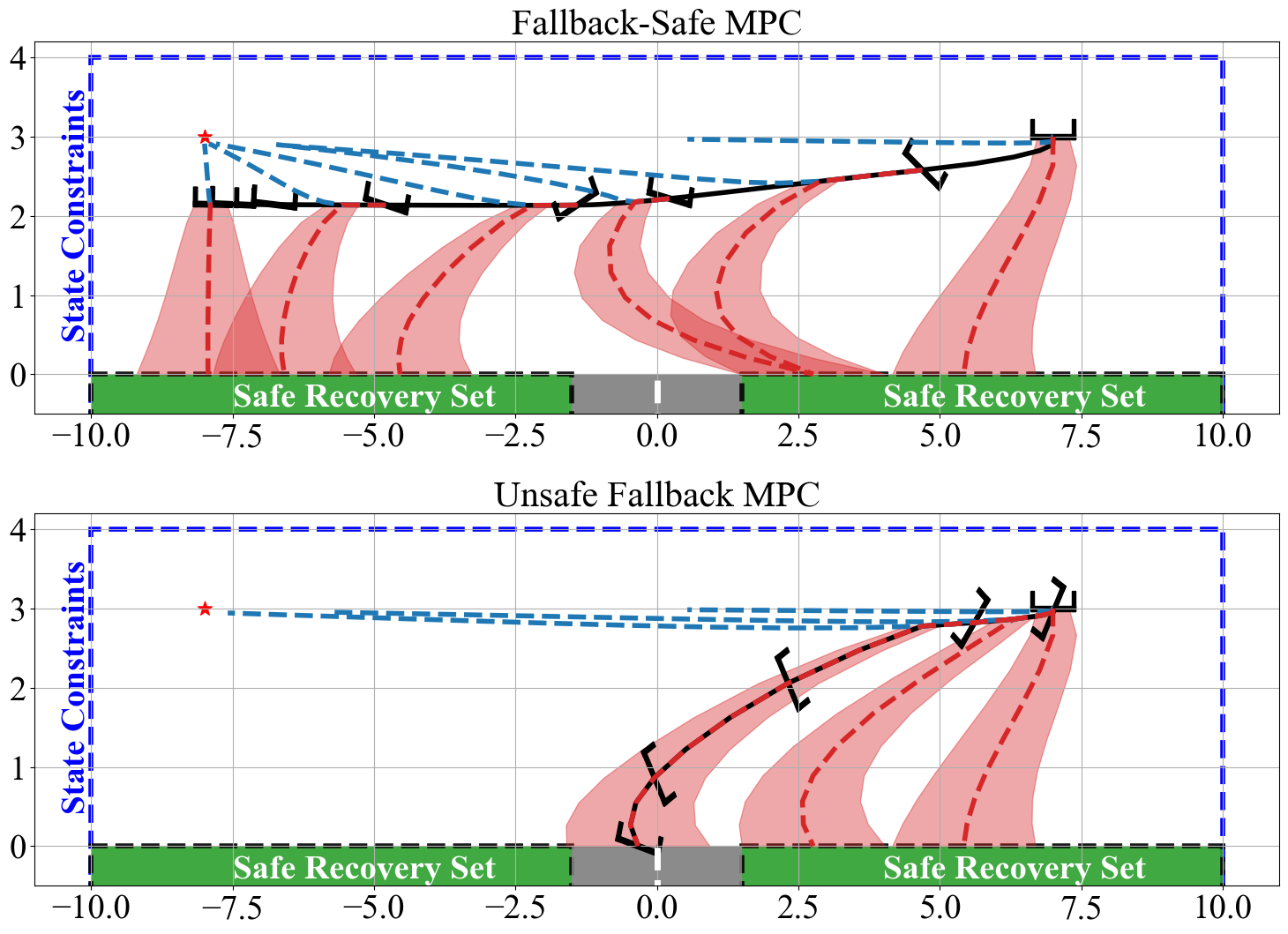}
    \caption{Trajectories of the planar quadrotor in the $xy$-plane. The disjoint safe recovery sets are highlighted in green, the unsafe ground region (e.g., a road), $|x|<1, \ y\leq0$, is in gray. Both the top and bottom figure follow the layout in \cref{fig:landing} (left).}
    \label{fig:goal_loc}
\end{figure}
Clearly, $\calX_R$ is a safe recovery set for $\pi_R(\y) = 0$, under the true dynamics \cref{eq:dynamics}.\footnote{\label{fn:infeas} We note that in this example, $\calX_R$ is not RPI under the state estimate dynamics \cref{eq:est-dyn} in nominal conditions, because the estimation error bound allows $\hat{\x}$ to leave the $\calX_R$ even if $\x \in \calX_R$. To retain the safety guarantee, we also trigger the fallback if \cref{eq:modification-mpc} is infeasible, a simple fix first proposed in \cite{KollerBerkenkampEtAl2018}. We did not observe recursive feasibility issues in the simulations.} For the drone to cross the road, we need to maintain recoverability with respect to either of the two disjoint recovery sets. We compare our approach with another naive baseline, which we label the Unsafe Fallback MPC, that executes a nominal MPC policy and naively tries to compute a fallback trajectory post-hoc, using the previous estimate before the fault occurred. As shown in \cref{fig:goal_loc}, our Fallback-Safe MPC first maintains feasibility of the fallback with respect to the rightmost recovery set, slows down, and then switches to the leftmost recovery set once a feasible trajectory crossing the road is found. In contrast, the Unsafe Fallback MPC does not maintain the feasibility of the fallback by modifying nominal operations, and is forced to crash land in the unsafe ground region (rather than throwing an infeasibility error, our implementation relies on slack variables). Therefore, this example illustrates that it is necessary to modify nominal operations to maintain the feasibility of a fallback. 

% In many use cases, maintaining recoverability with respect to a single recovery region might mean that we can never achieve the task. For example, in the drone delivery setting, recoverable regions are carved up and separated by the road network. If we maintain recoverability w.r.t. only one of these regions, then we will never cross the road. Instead, we need to maintain recoverability with respect to multiple disjoint recovery regions. This can be done by adding a mixed integer constraint to \cref{eq:modification-mpc} to select which region to recover to, or by simply solving multiple versions of \cref{eq:modification-mpc} at each timestep, one for each nearby recovery region, and selecting the input with the lowest associated cost.

\begin{figure}[t]
    \centering
    \includegraphics[width=\textwidth]{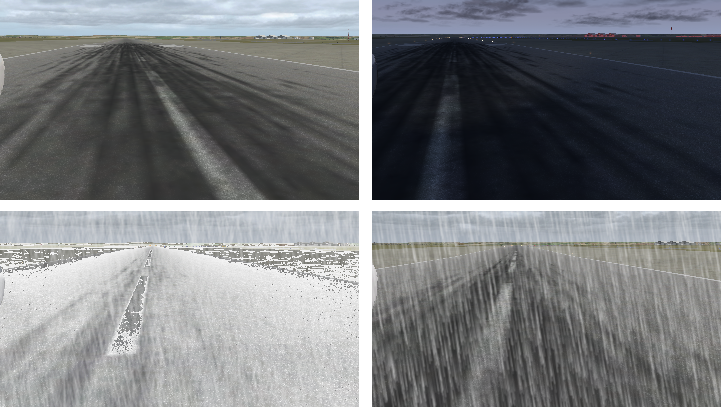}
    \caption{Simulated environments in the X-Plane 11 simulator. Top left: Morning, no weather. Top right: Night, no weather. Bottom left: Afternoon, snowing. Bottom right: Afternoon, raining.}
    \label{fig:xplane}
\end{figure}
\begin{figure}[b]
    \centering
    \includegraphics[width=\textwidth]{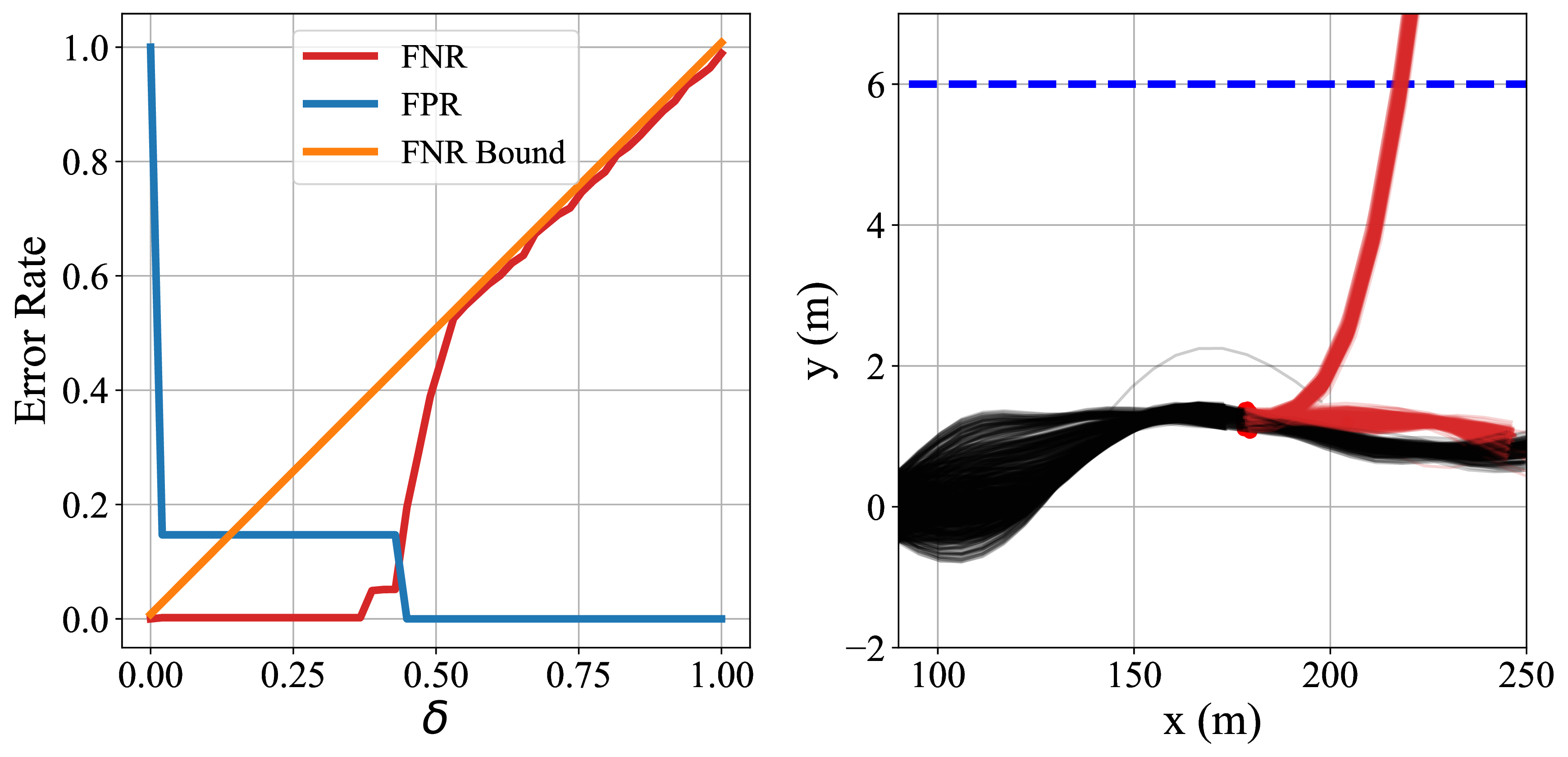}
    \caption{Left: FNR Bound indicates the $\delta + 1 / (|\calA| + 1)$ bound on the FNR from \cref{lem:conformal}. The FPR indicates the empirical rate at which we trigger the fallback without a perception fault ever occurring. The FNR indicates the empirical rate at which a perception fault occurs before we trigger the fallback. Right: closed-loop trajectories of the aircraft are in black. For trajectories in which the fallback triggered, we plot a red dot where the aircraft stopped, and plot in red the trajectory that would have occurred had we not triggered the fallback. }
    \label{fig:xplane-results}
\end{figure}

\textbf{X-Plane Aircraft Simulator:} Finally, we evaluate the conformal prediction Algorithm \ref{alg:conformal} and the end-to-end safety guarantee of our framework using the photo-realistic X-Plane 11 simulator. We simulate an autonomous aircraft taxiing down a runway with constant reference velocity, while using a DNN to estimate its heading error (HE) $\theta$ and center-line distance (cross-track error (CTE)) $y$ from an outboard camera feed. Here, internal encoders always correctly output the velocity $v$, so that $\x = [y, \theta, v]$ and $\y = v$. The aircraft must not leave the runway, given by the state constraint $|y|\leq 6 \mathrm{m}$. 

We train the DNN perception model on $4\times10^4$ labeled images collected only in morning, clear sky weather, but we deploy the system in a context $P_\rho$ where it may experience a variety of weather conditions (depicted in \cref{fig:xplane}). We parameterize the environment $\rho := (\mathrm{weather \ type}, t_{\mathrm{start}}, \mathrm{severity})$ as a triplet indicating the weather type, severity level, and starting time from which the visibility starts to degrade, so that under the environment distribution $P_{\rho}$, we randomly sample an environment that starts with clear-skies and high visibility, but may cause OOD errors during an episode. As shown in \cref{fig:example}, heavy weather degrades the perception significantly. The fallback is to brake the aircraft to a stop, where the stopped states are invariant\textsuperscript{\ref{fn:infeas}} under $\pi_R(\y) := 0$. As in \cite{RichterRoy2017}, we train an autoencoder alongside the DNN on the morning, clear sky data and use the reconstruction error as the anomaly signal $a_t$. We define the perception error set to include at most $7$ degree HE and $1.3\mathrm{m}$ CTE.

We record 100 training trajectories using the Fallback-Safe MPC \cref{eq:modification-mpc} and a ground-truth supervisor $w(\cdot) = 1 - \mathbf{1}\{\pe_t \in \calE\}$ to calibrate Algorithm \ref{alg:conformal}, and then evaluate on 900 test trajectories with environments sampled i.i.d. from $P_\rho$ using and Algorithms \ref{alg:fallback-safe-mpc}-\ref{alg:conformal}. This ensures we satisfy \cref{def:safety} by \cref{cor:endtoend}. We compute the empirical false positive and false negative rate when we evaluate Algorithm \ref{alg:conformal} with various values of $\delta \in [0,1]$ in \cref{fig:xplane-results} (left). As \cref{fig:xplane-results} (left) shows, the FNR of Algorithm \ref{alg:conformal} satisfies \cref{lem:conformal}'s $\delta' = \delta + 1 / (|\calA| + 1)$ bound for all values of $\delta$, validating our guarantees. Moreover, the false positive rate, the rate at which we incorrectly trigger the fallback, is near $0$ for risk tolerances as small as $\delta' = 5\%$. This shows that algorithm \ref{alg:conformal} is highly sample efficient and not overly conservative, since it hardly issues incorrect alarms with orders of magnitudes fewer samples than we needed to train the perception. In \cref{fig:xplane-results} (right), we control the system with an end-to-end safety guarantee of $\delta' = .1$ using our framework and observe no constraint violations. For the trajectories in which we triggered the fallback, \cref{fig:xplane-results} (right) shows that over $80\%$ would have led to an aircraft failure had we not interfered. This shows that our framework is effective at avoiding robot failures, and experiences few unnecessary interruptions with an effective OOD detection heuristic like the autoencoder reconstruction loss.
%% 1. use algorithm 1
%% make a guarantee using corrolary end-to-end
%% Xplane 11 simulator
%% Aircraft taxiing down a runway
%% Cross-track error < 6m
%% DNN is trained on morning, clear sky weather. Evaluate on mixture of no weather, and sudden shifts to darkness, rain, snow, which derail the DNN predictions
%% Use an autoencoder reconstruction loss to do OOD detection
%% The fallback is to brake, and stop the aircraft as soon as possible. 
%% 
%% record 20 training trajectories {replace with 200}, and evaluate the predictor on 80 {replace with 1800} test trajectories. time limit is 60 seconds
%% 
\section{Conclusion}
In this work, we have formalized the design of safety-preserving fallback strategies under perception failures by ensuring the feasibility of a fallback plan with respect to a safe recovery set, a subset of the state space that we can make invariant without full knowledge of the state. Similar to the terminal invariant in a standard MPC, we have demonstrated that recovery sets can readily be identified offline. Our simulations also showed that the calibration procedure, which enables strong safety assurances, is particularly amenable to limited data collection pre-deployment.
% The simulated quadrotor examples demonstrate that we can readily identify recovery regions in practical examples. Moreover, our vision-based aircraft simulations show that our calibration procedure enables us to make strong safety assurances with orders of magnitudes fewer samples than required to retrain the perception model when deploying in a new context (e.g., we avoided all failures with $100$ calibration samples but trained the DNNs on $\approx10^4$ training samples), thus making our approach particularly attractive when a limited budget is available for pre-deployment data collection.
Still, we observe that our runtime monitor occasionally triggers the fallback when the closed-loop system would not have violated safety constraints because we rely on an imperfect heuristic for OOD detection. Therefore, future work should investigate how to tune runtime monitors to only detect downstream failures more effectively. In addition, future work may explore more complex statistical analysis on the runtime monitor since our framework currently does not permit a switch back to nominal operations after a fault occurs. 
%%%%%%%%%%%%%%%%%%%%%%%%%%%%%%%%%%%%%%%%%%%%%%%%%%%%%%%%%%%%%%%%%%%%%%%%%%%%%%%%
% \section{ACKNOWLEDGMENTS}

% Scratch math for matteo:

% \vspace{\baselineskip}

% $\hat{\x}_t$ 

% \vspace{\baselineskip}
% \\
% $w(a_t) = 1$

% \vspace{\baselineskip}
% \\
% $w(a_t) = 0$

% \vspace{\baselineskip}
% \\
% \begin{align*}
%     \fu^\star = \argmin_{\fu}& \ C(\hat{\x}, \fu) \\ 
%     \mathrm{s.t.} &\quad \x^+ = \f(\x, \fu)\dots
% \end{align*}

% \vspace{\baselineskip}

%%%%%%%%%%%%%%%%%%%%%%%%%%%%%%%%%%%%%%%%%%%%%%%%%%%%%%%%%%%%%%%%%%%%%%%%%%%%%%%%

\bibliographystyle{IEEEtran-short}
\bibliography{references}

\onecolumn
\begin{appendices}
\crefalias{section}{appendix}
\setcounter{section}{0}
\section{Glossary}\label{ap:glos}

A glossary of all the notational conventions and symbols used in this paper is included in Table \ref{tab:glos}.
\begin{table}[!h]
\adjustbox{height=10cm, center}{
\begin{tabularx}{\textwidth}{c c X}
\hline
   - & Symbol & Description \\
\hline
   & $x$ & unless explicitly defined otherwise, scalar variables are lowercase \\
   & $\mathbf{x}$ & vectors are boldfaced \\
   & $\calX$ & sets are caligraphic \\ 
   & $x_t$ & time-varying quantities are indexed with a subscript $t \in \mathbb{N}_{\geq 0}$\\ 
  & $\x_{0:t}$ & Shorthand to index subsequences: $\x_{0:t} := \{\x_0,\dots, \x_t\}$ \\
   & $\mathbf{R}$ & matrices are uppercase and boldfaced\\
   & $\prob(x \geq 0)$ & probability of the event $x \geq 0$ \\
  Conventions and Notation & $\prob(\calA \ | \ \calB)$ & probability of the event $\calA$ conditioned on the event $\calB$ \\
   & $\lambda, \delta, \epsilon, \theta$ & hyperparameters (regardless of their type) are lowercase Greek characters \\
   & $\x_{t+k|t}$ & Predicted quantities at $k$ time steps into the future computed at time step $t$. Read $x_{t+k|t}$ as ``the predicted value of $x$ at time $t + k$ given time $t$.'' \\
   & $\calX \oplus \calY$ & Minkowski sum operation, defined as $\calX \oplus \calY := \{\x + \y \ : \ \x \in \calX, \ \y \in \calY\}$ \\
   & $\calX \ominus \calY$ & Pontryagin set difference, defined as $\calX \ominus \calY := \{\x \ : \ \x + \y \in \calX, \ \forall \y \in \calY\}$ \\
   & $\mathrm{TV}(P,Q)$ & Total variation distance between distributions $P$ and $Q$\\
   & $\mathbf{1}\{x\}$ & Indicator function, $\mathbf{1}\{\mathrm{True}\} := 1$, $\mathbf{1}\{\mathrm{False}\} := 0$. \\ 
   & $\f(\calX)$ & Shorthand notation to denote the image of the input set $\calX$ under the mapping $\f$. That is, $\f(\calX) := \{\y = \f(\x) \ : \ \x \in \calX\}$ \\
   & $\mathbf{A} \calX$ & Shorthand notation to denote the linear transformation of the elements of $\calX$ with the matrix $\mathbf{A}$.\\ 
    \hline
    \hline
    & $\x$ & System state \\
    & $\fu$ & Input \\
    & $\w$ & Disturbance \\
    & $\f$ & Dynamics \\
    & $\z$ & Perception observation \\
    & $\rho$ & Environment parameter \\
    & $p_\rho$ & Process from which observations and disturbances are sampled during a trajectory \\
    & $P_\rho$& Environment context: The distribution over  environments. \\
    & $\calX$ & State constraint set \\
    & $\calU$ & Input constraint set \\
    & $t_{\mathrm{lim}}$ & Time limit on a trajectory (episode/iteration length) \\
    & $\delta$ & Risk tolerance \\
    & $\hat{\x}$& Perception state estimate \\
    &$\y$& Fallback measurement \\
    &$\g$& Fallback measurement function\\
    Variables &$\calY$& Fallback measurement output set \\
    &$a$& Anomaly/OOD detection signal \\
    &$\pe$& Perception error $\x - \hat{\x}$ \\
    &$\pi_R$& Recovery policy \\
    & $\calX_R$ & Recovery set \\
    &$ \calE_\theta$ & Perception error set \\
    & $\theta$ & Hyperparameter defining the perception error set \\
    & $T$& MPC horizon \\
    & $\fu^F$ & Fallback policy \\
    & $\calP_N$ & Nominal policy set \\
    & $\calP_F$ & Fallback policy set \\
    & $\calR_k(\hat{\x}_t, \fu^F_{0:T})$ & $k-$step reachable set of the true state at $\hat{\x}_t$ under fallback policy $\fu^F_{0:T}$ \\
    & $\widehat \calR_k(\hat{\x}_t, \fu^F_{0:T})$ & $k-$step reachable set of the estimate at $\hat{\x}_t$ under fallback policy $\fu^F_{0:T}$  assuming nominal operation\\
    & $C$ & MPC cost function \\
    & $\fu^\star$ & The starred superscript denotes an optimal input to \cref{eq:modification-mpc}. \\
    & $w$ & Runtime monitor \\
    & $\tau$& Trajectory \\
    &$\calT$& Set of trajectories ($(\calX \times \calU \times \calX \times \R) ^{t_{\mathrm{lim}}}$)\\
    & $\calD$ & Trajectory calibration dataset \\
    & $t_{\mathrm{stop}}$ & Stopping time of a trajectory \\
    & $P$ & Distribution of a trajectory in the calibration dataset \\
    & $P'$ & Distribution of a trajectory at deployment under Algorithm \ref{alg:fallback-safe-mpc} and Algorithm \ref{alg:conformal}. \\
    & $P_{t_{\mathrm{stop}}|\mathrm{fault}}$ & the distribution of $(\x, \hat{\x}, a)$ at $t_{\mathrm{stop}}$ conditioned on the event that $\pe_{t_{\mathrm{stop}}} \not \in \calE_\theta$ under a trajectory sampled from $P$ \\
    & $P_{t_{\mathrm{stop}}|\mathrm{fault}}'$ & the distribution of $(\x, \hat{\x}, a)$ at $t_{\mathrm{stop}}$ under a trajectory sampled from $P'$,  conditioned on both $\pe_{t_{\mathrm{stop}}} \not \in \calE_\theta$ and  $t_{\mathrm{fail}}\geq t_{\mathrm{stop}}$ \\
    \hline
\end{tabularx}
}
\label{tab:glos}
\caption{Glossary of notation and symbols used in this paper.}
\end{table}

\section{Extended Related Works}\label{ap:relwork}
The fact that modern deep learning models behave poorly on OOD data has been extensively documented in recent years. For perception algorithms, this results in models failing in ways that are difficult to anticipate because it is impossible to derive or intuit what makes the pixel values of high-dimensional input image ``different'' to training data from first-principles \cite{GeirhosJacobsenEtAl2020, TorralbaEfros2011, RechtRoelofsEtAl2019}. For example, researchers have shown that vision models often confidently make incorrect classifications on images with unseen classes \cite{OvadiaFertigEtAl2019} or minor corruptions or perturbations \cite{HendrycksDietterich2019} and that they greatly deteriorate in performance when the input domain subtly changes due to e.g., differences in lighting (such as day-night shifts), weather, and background scenery (e.g., across neighboring countries or even between rural and urban scenes) \cite{SharmaAzizanEtAl2021,GeirhosJacobsenEtAl2020, KohSagawaEtAl2021}. In the ML community, this has primarily been studied from a model-centric view by constructing benchmarks to study model degradation and by proposing algorithms that improve domain generalization and domain adaptation algorithms that adapt models in response to data from shifted distributions (see \cite{SinhaSharmaEtAl2022}, and the references therein for an overview). However, few platforms exist to benchmark the degradation of a closed-loop control system that repeatedly applies a learned model in feedback. Moreover, it is a basic fact that we cannot anticipate or robustify against all OOD failure modes \cite{SinhaSharmaEtAl2022, SeshiaSadighEtAl2016}. Therefore, we implement several common OOD scenarios in the photorealistic X-Plane simulator to test our approach, and open-source our benchmark platform as a community resource. 

 In addition, a large number of OOD detection and runtime monitoring algorithms have been proposed to detect when a perception model is operating outside of its competence \cite{RahmanCorkeEtAl2021}. These methods commonly train an additional network to flag inputs that are dissimilar from the training data in a variety of ways, for example, by directly modeling the input distribution, measuring distances in latent-spaces, analyzing autoencoder reconstruction errors, or directly classifying whether inputs are OOD with one-class classification losses (see \cite{SalehiMirzaeiEtAl2021, RuffKauffmanEtAl2021, YangZhouEtAl2021} for an overview). Other methods construct improved uncertainty scores produced by the perception model by approximating the Bayesian posterior through e.g., ensembling \cite{LakshminarayananPritzelEtAll2017}, Monte-Carlo dropout \cite{GalZoubin2016}, choices of architecture or loss functions \cite{AbdarPourpanahEtAl2021, AminiSchwartingEtAl2020, OsbandWenEtAl2023, LiuLinEtAl2020}, or LaPlace approximations \cite{SharmaAzizanEtAl2021}. In general, these algorithms are heuristics that correlate with perception faults, but do not provide formal guarantees of correctness. In fact, recent work proved that the OOD detection problem is not PAC (provably approximately correct) learnable in a number of settings \cite{FangLiEtAl2022}. Moreover, even algorithms that certifiably detect perception faults in certain settings (e.g., by checking consistency in outputs between sensing modalities and across time e.g., see \cite{AntonanteSpivakEtAl2021, RahmanCorkeEtAl2021}) first require us to define what constitutes a perception fault (i.e., ``how bad is too bad''), because state estimates are never exactly correct. Therefore, we jointly design the runtime monitor and control stack by specifying an error bound on nominal perception errors in the control design, and then construct a runtime monitor that uses an OOD detector to predict when that bound is violated.

We take this approach because it is difficult to apply existing work on output-feedback to a perception-enabled system. In particular, a significant body of work constructs controllers that satisfy state and input constraints for all time despite partial observability using robust model predictive control (MPC) (e.g., see \cite{MayneRakovic2006, LorenzettiPavone2020, LovaasSeronEtAl2008,KohlerMuller2021,FindeisenImslandEtAl2003,GoulartKerrigan2006, ManzanoLimonEtAl2020}). A common approach is use knowledge of the dynamics and measurement model with assumptions on system observability to construct state estimators that persistently satisfy known error bounds. They then robustly control the state estimate dynamics using a robust MPC algorithm, tightening constraints to account for the estimation error \cite{MayneRakovic2006, LorenzettiPavone2020, KohlerMuller2021, GoulartKerrigan2006}. However, we cannot model high-dimensional measurements like images from first principles, and as a result, we must rely on black-box neural networks to extract scene information, which may exhibit arbitrarily poor error behavior on OOD inputs. Moreover, in our setting, the measurements that remain reliable after the perception fails are insufficient to recover the state (e.g., i.e., fallback system is not observable). We rely on the ideas of output-feedback MPC to nominally control the robot, but maintain feasibility to a backup plan in case the perception unexpectedly degrades.

In addition, our approach takes inspiration from existing work on \emph{safety filters} that modify an ML-based black-box policy's actions to maintain safety \cite{BrunkeGreefEtAl2022}. Such methods minimally modify nominal decisions to ensure invariance of a safe region of the state space when the black-box policy would take actions to leave that set. Many such algorithms have been proposed in recent years, for example by using robust MPC \cite{WabersichZeilinger2018}, control barrier functions (CBFs) \cite{ChengOrosz2019} and Hamilton-Jacobi reachability \cite{FisacAkametaluEtAl2019, LeungSchmerlingEtAl2019}. However, our method differs from such approaches in two important ways: First, existing safety filters typically operate under assumptions of perfect state knowledge, or that an estimate of sufficient accuracy is always available \cite{BrunkeZhou2022}. In contrast, we must ensure the fallback strategy satisfies full state constraints only using the last accurate estimate and any remaining partial state information. To account for these discrete information modes, our insight is that we can often specify \emph{recovery regions} to fallback into, safe regions of the state space that we can make invariant without full knowledge of the state. For example, the drone in \cref{fig:hero-fig} does not need an obstacle detector to avoid collisions when it is landed in a field. Secondly, existing safety filters take a ``zero-confidence'' view in black-box learned components: They continually ensure the safe operation of the system by aligning an ML-enabled controller's output with those of a backup policy that never uses ML model outputs. We consider a setting where ML perception is critical to achieve the task, and therefore such an approach is much too conservative. Instead, we recognize that ML-enabled components are generally reliable, but leverage OOD detection to transition to a fallback strategy in rare failure modes.

Rather than filtering a black-box learned system's outputs to remain safe, another major line of research has been to construct controllers that ensure learned models operate within their domain of competence \cite{BrunkeZhou2022}. This includes work on policy optimization in (model-based, offline) reinforcement learning (RL) \cite{SchulmanLevineEtAl2015, LevineKumar2020}, where model updates are typically constrained to trust regions to keep trajectories close to existing data, and learning-based control \cite{BrunkeZhou2022, HewingKabzan2020}, where MPC planning is typically combined with Bayesian model-learning to avoid states with high uncertainty in the dynamics. RL algorithms have been extensively applied to vision-based tasks \cite{SchulmanLevineEtAl2015, LevineFinn2016, ThanajeyanBalakrishna2021}. In addition, some recent approaches in learning-based control propose to learn the error behavior of a vision system as a function of the state and then verify closed-loop properties \cite{KatzCorso2022} or robustly plan while taking these error bounds into account \cite{DeanMatni209, ChouOzay2023, IchterLandry2020}, for example, by making smoothness assumptions on the vision's error behavior \cite{DeanMatni209}. However, these algorithms require the environment to remain fixed (i.e., that the mapping from state to image is constant), such that degradation in model quality is solely a result of visiting unseen states due to changes in the control policy. Instead, many of the subtle failure modes of perception systems are caused by environmental changes beyond the control of the robot, like weather changes. Therefore, the robot cannot act to retain confidence in the perception system, so we consider triggering a fallback the only reasonable alternative.

We take inspiration from existing work on fault-tolerant control that maintains feasibility of passive-backup \cite{GuffantiAmico2023}, abort-safe \cite{MarsillachCairano2020}, or contingency plans \cite{AlsterdaBrownEtAl2019} under actuation or sensor failure using MPC, we modify the nominal operation to ensure the existence of a safety preserving fallback. However, in many such works, it is assumed that faults are perfectly detected and that these systems function perfectly nominally. It is challenging to detect failures in ML-based systems and, as illustrated in \cref{fig:example}, errors are nominally tolerable, but nonzero. Therefore, we jointly design the control stack and runtime monitor to account for such errors. Most closely related to our approach are several applied works that trigger a fallback controller by thresholding an OOD detector in the robotics field \cite{RichterRoy2017, FilosTigas2020, McAllisterKahn2019}. These works use fallbacks that are domain-specific or assumed to be safe, assume the ML models function perfectly nominally, and rely on OOD detectors without end-to-end guarantees of system safety. We formalize the design of fallback strategy through the definition of recovery sets and make end-to-end guarantees.

To make an end-to-end guarantee, we leverage recent results in conformal prediction. This is because traditional methods for fault detection, such as model-based approaches, are difficult to apply to vision failures. Instead, we learn how to rely on a heuristic OOD detector through a conformal inference procedure (see \cite{AngelopoulosBates2022} for an overview). Conformal methods are attractive in this setting because they produce strong guarantees on the correctness of predictions with very little data and those guarantees are distribution-free -- that is, they do not depend on assumptions on the data-generating distribution \cite{AngelopoulosBates2022, BalasubramanianHoEtAll2014}. However, such methods have not generally been proposed for making high-probability guarantees jointly for all time in a sequential prediction setting, where inputs are correlated over time. While some recent work has aimed to move beyond the i.i.d. setting \cite{BarberCandes2023, TibshiraniBarber2019}, or making sequentially valid predictions across i.i.d. trajectories \cite{LuoSinha2023}, these cannot be applied sequentially over correlated observations within a trajectory. Instead, we adapt an existing algorithm, \cite{LuoZhao2023}, to the sequential prediction setting by noticing we only require a guarantee on the first detection of a fault to ensure the end-to-end safety of our framework.

\section{Proof of \cref{lem:reachable}}\label{ap:reachable}
\reachsets*

\begin{proof}
    First, we prove that $\hat{\x}_{t+k} \in \widehat{\calR}_k$ for all $k\in \{0, \dots T+1\}$. We proceed by induction.  Here, we drop arguments to $\calR$ for notational simplicity (i.e., $\calR_0 := \calR_0(\hat{\x}_t, \fu_{0:T}^F)$). As the base case, note that $\hat{\x}_t \in \widehat\calR_0$ by definition.
    % , so we have that $\hat{\x}_t, \ \x_t \in\calR_0$ since we assume $\pe_t \in \calE_\theta$ and $0 \in \calE_\theta$. 
    For the inductive step, assume $\hat{\x}_{t+k} \in \widehat{\calR}_{k}$ for some $k \in \{0, \dots, T\}$.
    % Then, since we assume $0, \  \pe_{t+k} \in \calE_\theta$, we also have that $\x_{t+k}, \ \hat{\x}_{t+k} \in \calR_k$. 
    Since 1) we assume $\pe_{t:t+T+1} \subset \calE_\theta$ and 2) the estimated state follows the dynamics in \cref{eq:est-dyn}, that is, that $\hat{\x}_{t+k+1} = \f(\hat{\x}_{t+k} + \pe_{t+k}, \fu_{k}^F(\g(\hat{\x}_{t+k} + \pe_{t+k})), \w_{t+k}) - \pe_{t+k+1}$, it holds that $\hat{\x}_{t+k+1} \in \widehat{\calR}_{k+1}$ by construction. Therefore, $\hat{\x}_{t+k} \in \widehat{\calR}_{t+k}$ for all $k \in \{0,\dots, T+1\}$.

    Next, we prove that $\widehat{\calR}_{k}(\hat{\x}_{t+1}, \fu_{1:T}^F) \subseteq \widehat{\calR}_{k+1}(\hat{\x}_t, \fu_{0:T}^F)$ for $k \in \{0, \dots, T\}$ by induction. As the base case, note that $\hat{\x}_{t+1} \in \hat{\calR}_{1}(\hat{\x}_t, \fu_{0:T}^F)$. Therefore, $ \{\hat{\x}_{t+1}\} = \widehat{\calR}_{0}(\hat{\x}_{t+1}, \fu_{1:T}^F) \subseteq \widehat{\calR}_{1}(\hat{\x}_t, \fu_{0:T}^F)$. For the inductive step, assume that $\widehat{\calR}_{k}(\hat{\x}_{t+1}, \fu_{1:T}^F) \subseteq \widehat{\calR}_{k+1}(\hat{\x}_t, \fu_{0:T}^F)$ for some $k \in \{0, \dots,  T - 1\}$. Suppose $\hat{\x}' \in \widehat \calR_{k+1}(\hat{\x}_{t+1}, \fu_{1:T}^F)$. Then, by definition, $\hat{\x}' = \f(\hat{\x} + \pe, \fu_{k+1}^F(\g(\hat{\x} + \pe)), \w) - \pe'$ for some $\hat{\x}\in \widehat\calR_{k}(\hat{\x}_{t+1}, \fu_{1:T}^F)$, $\w \in \calW$,  and $\pe, \ \pe' \in \calE_\theta$. This immediately implies that $\hat{\x}' \in \widehat\calR_{k+2}(\hat{\x}_{t}, \fu_{0:T}^F)$, since we assume $\widehat\calR_{k}(\hat{\x}_{t+1}, \fu_{1:T}^F) \subset \widehat\calR_{k+1}(\hat{\x}_{t}, \fu_{0:T}^F)$. Therefore, $\widehat \calR_{k+1}(\hat{\x}_{t+1}, \fu_{1:T}^F) \subseteq \widehat\calR_{k+2}(\hat{\x}_{t}, \fu_{0:T}^F)$. By induction, we therefore have that 
    \begin{equation}\label{eq:proof-inclusion}
        \widehat{\calR}_{k}(\hat{\x}_{t+1}, \fu_{1:T}^F) \subseteq \widehat{\calR}_{k+1}(\hat{\x}_t, \fu_{0:T}^F), \quad \forall k \in \{0, \dots, T\}.
    \end{equation}

    Finally, \cref{eq:proof-inclusion} immediately implies that $\calR_k(\hat{\x}_{t+1}, \fu_{1:T}^F) \subseteq \calR_{k+1}(\hat{\x}_t, \fu_{0:T}^F)$ for all $k \in \{0, \dots, T\}$. Moreover, since $0 \in \calE_\theta$, we have that $\widehat\calR_k(\hat{\x}_t, \fu_{0:T}^F) \subseteq \calR_k(\hat{\x}_t, \fu_{0:T}^F)$ for $k \in \{0, \dots, T+1\}$. In addition, since 1) we proved $\hat{x}_{t+k} \in \widehat\calR_k(\hat{\x}_t, \fu_{0:T}^F)$ and 2) we assume $\pe_{t:t+T+1}\subset \calE_\theta$, we have that $\x_{t+k} \in \calR_k(\hat{\x}_t, \fu_{0:T}^F)$ for all $k \in \{0, \dots, T+1\}$.
    %since $0\in \calE_\theta$, this immediately implies $\widehat\calR_k \subseteq \calR_k$ for all $k \in \$
    % \vspace{\baselineskip}
    % We proceed by induction. As the base case, note that $0\in \calE_{\theta}$ implies $\hat{\x}_t \in \calR_0$ and that $\x_t = \hat{\x}_t + \pe_t \in \calR_0$ since we assume $\pe_t \in \calE_\theta$. For the inductive step, assume $\hat{\x}_{t+k}, \x_{t+k} \in \calR_{k}$ for $k \in \{0, \dots, T\}$. Then, note that $\x_{t+k+1} = \f(\x_{t+k}, \fu_{k}(\g(\x_{t+k})), \w_{t+k}) \in \calR_{k+1}$ since $\x_{t+k} \in \calR_k$, $\w_{t+k} \in \calW$, and $0 \in \calE_\theta$. Finally, this also implies that $\hat{\x}_{t+k+1} = \x_{t+k+1} - \pe_{t+k+1} \in \calR_{k+1}$, since $\calE_{\theta}$ is symmetric (that is, $\pe \in \calE_{\theta} \implies -\pe \in \calE_\theta$).
\end{proof}

\section{Proof of \cref{thm:mpc-safety}}\label{ap:mpc-safety}
As a shorthand in the theorem statement and the following proofs, we use the notation $t_{\mathrm{fail}}$ to denote the first time step that the runtime monitor raises an alarm, i.e., that $w(a_t) = 0$ for all $ t < t_{\mathrm{fail}}$ and that $w(a_{t_{\mathrm{fail}}}) = 1$ if $t_\mathrm{fail} < \infty$.
To prove \cref{thm:mpc-safety}, we first prove the recursive feasibility of the MPC \cref{eq:modification-mpc} in Algorithm \ref{alg:conformal} up to the failure time $t_{\mathrm{fail}}$. 

\begin{lemma}\label{lem:recfeas}
    Consider the closed-loop system formed by the dynamics \cref{eq:dynamics} and the Fallback-Safe MPC (Algorithm \ref{alg:fallback-safe-mpc}). Suppose that $\pi_R \in \calP_F$, and that the runtime monitor $w$ does not miss a detection of a perception fault, i.e., that $\x_t - \hat{\x}_t \in \calE_\theta$ for all $t < t_{\mathrm{fail}}$. Then, if the Fallback-Safe MPC problem \cref{eq:modification-mpc} is feasible at $t = 0$ and $w(a_0) = 0$, we have that 1) the MPC problem \cref{eq:modification-mpc} is feasible for all $t < t_{\mathrm{fail}}$ and that 2) the closed-loop system  satisfies $\x_t \in \calX$, $\fu_t \in \calU$ for all $t < t_{\mathrm{fail}}$.
\end{lemma}
\begin{proof}
Since we assume $w(a_0) = 0$, we have that $t_{\mathrm{fail}} \geq 1$. Therefore, suppose the MPC \cref{eq:modification-mpc} is feasible at some time $t \in \{0, \dots, t_{\mathrm{fail}} - 1\}$ with optimal fallback policy sequence $\fu_{t:t+T|t}^{F,\star}$. Then, $\x_t \in \calX$, since $\x_t \in \calR_{0}(\hat{\x}, \fu_{t:t+T|t}^{F,\star}) \subseteq \calX$ by \cref{lem:reachable} and the assumption that $\pe_{0:t_{\mathrm{fail}} - 1} \subseteq \calE_{\theta}$. In addition, we have that Algorithm \ref{alg:fallback-safe-mpc} applies the input $\pi(\hat{\x}_t, \y_t) = \fu_{t|t}^\star(\hat{\x}_t) \in \calU$, since we enforce $\fu_{t|t}^\star(\hat{\x}_t) = \fu^{F,\star}_{t|t}(\y_t) $ and $\fu^{F,\star}_{t|t}(\y_t) \in \calU$ by construction in \cref{eq:modification-mpc}.

Next, consider the candidate fallback policy sequence $\fu_{t+1:t+T+1|t}^F := \{\fu_{t+1|t}^{F,\star}, \dots, \fu_{t+T|t}^{F,\star}, \pi_R\}$ for problem \cref{eq:modification-mpc} at time $t+1$. If $t+1 < t_{\mathrm{fail}}$, \cref{lem:reachable} gives us
% the sets $\calR_{t+1:t+T+1}(\hat{\x}_t, \fu_{t:t+T}^{F,\star})$ contain all possible realizations of $\x_{t+1:t+T+1}$ and $\hat{\x}_{t+1:t+T+1}$ under $\fu_{t+1:t+T+1:t}^F$, and 
that $\calR_{k}(\hat{\x}_{t+1},\fu_{t+1:t+T|t}^F) \subseteq \calR_{k+1}(\hat{\x}_{t},\fu_{t:t+T|t}^{F,\star})$ for $k \in \{0, \dots, T \}$, since we assume $\pe_{0:t_{\mathrm{fail}-1}} \subset \calE_\theta$. 
Moreover, since $\calR_{T+1}(\hat{\x}_{t}, \fu_{t:t+T}^{F,\star}) \subseteq \calX_R$, we have that $\calR_{T+2}(\hat{\x}_{t+1}, \fu_{t+1:t+T+1|t}^F) \subseteq \calX_R$ by \cref{as:recovery} and \cref{lem:reachable}. Therefore, the candidate fallback policy sequence $\fu_{t+1:t+T+1|t}^F$ satisfies the state and input constraints in \cref{eq:modification-mpc} at $t+1$. Since we assume there always exists a $\fu \in \calP_N$ such that $\fu(\hat{\x}_{t+1}) = \fu_{t+1|t}^{F,\star}(\y_{t+1})$, problem \cref{eq:modification-mpc} is therefore feasible at time $t+1$. The lemma then holds by induction, since we assume \cref{eq:modification-mpc} is feasible at $t=0$ and that $w(a_0) = 0$.
% Suppose the MPC \cref{eq:modification-mpc} is feasible at some time $t \in \{0,\dots, t_{\mathrm{fail}} - 2\}$. Then, applying the control input ensures that both the state and the next measurement satisfy state constraints at the next timestep because $\fu_{t|t} = \fu_{t|t}^R$. Then, note that the optimal input sequence chained with the known recovery policy guarantees there exists a sequence of inputs that satisfy the state constraints and reach $\calX_R$ at time $t+1$, so the MPC is feasible at $t+1$, no matter the actual realization of $\pe_t$. The theorem then holds by induction.
\end{proof}

\vspace{\baselineskip}
The proof of \cref{thm:endtoend} then follows by combining \cref{lem:recfeas} with the logic that triggers the fallback in  Algorithm \ref{alg:fallback-safe-mpc}.

\vspace{\baselineskip}
\textbf{Proof of \cref{thm:mpc-safety}:}

\fallbacksafe*

\begin{proof}
    By \cref{lem:recfeas}, feasibility of the MPC \cref{eq:modification-mpc} at $t=0$ and $w(a_0) = 0$ implies that the MPC is feasible for all $0 \leq t \leq t_{\mathrm{fail}} - 1$ and that the system satisfies state and input constraints for $t <t_{\mathrm{fail}}$. Moreover, since a) \cref{eq:modification-mpc} is feasible at time step $ t_{\mathrm{fail}} - 1$, b) $\fu_{t_{\mathrm{fail}} - 1|t_\mathrm{fail} - 1}^\star(\hat{\x}_{t_\mathrm{fail} - 1}) = \fu_{t_{\mathrm{fail}} - 1|t_\mathrm{fail} - 1}^{F,\star}(\y_{t_\mathrm{fail} - 1})$ by construction, and c) Algorithm \ref{alg:fallback-safe-mpc} applies the fallback strategy $\fu_{t|t_{\mathrm{fail}} - 1}^{F,\star}$ for $t \in \{t_{\mathrm{fail}}, \dots, t_{\mathrm{fail}} +T - 1\}$, it holds that $\x_{t} \in \calX$ and $\fu_t \in \calU$ for all $t \in \{t_{\mathrm{fail}}, \dots, t_{\mathrm{fail}} +T - 1\}$ by \cref{lem:reachable}. Moreover, feasibility of the MPC \cref{eq:modification-mpc} at $ t_{\mathrm{fail}} - 1$ also implies that $\x_{t_{\mathrm{fail}} +T } \in \calX_R$. Furthermore, \cref{as:recovery} and $0 \in \calE_\theta$ give us that the application of $\pi_R$ for all $t\geq t_{\mathrm{fail}} +T  $ ensures that $\x_t \in \calX_R \subseteq \calX$ and $\fu_t \in \calU$ for all $t\geq t_{\mathrm{fail}} +T$.
    % Todo. The proof is basically trivial, relying on \cref{thm:mpc-safety}, feasibility of \cref{eq:modification-mpc} for an open loop sequence to reach $\calX_R$, and \cref{as:recovery}
\end{proof}

\section{Tractable reformulation of \cref{eq:modification-mpc} for linear-quadratic systems}\label{ap:mpc-linear}
We consider linear systems of the form
\begin{equation}
\begin{aligned}
\x_{t+1} &= \mathbf A\x_t +\mathbf B \fu_t + \w_t \\
\y_t &= \mathbf C\x_t
\end{aligned},
\end{equation}
subject to polytopic constraints on states and inputs, 

\begin{equation*}
    \calX  = \{\x \ : \ \mathbf H_x \x \leq  \mathbf h_x\}, \quad \calU = \{\fu \ : \ \mathbf H_u \fu \leq \h_u\},
\end{equation*}
and a polytopic disturbance and perception error bound of the form
\begin{equation*}
    \calW = \{\w \ : \ \mathbf H_w \w \leq \mathbf h_w\}, \quad  \calE_{\theta} = \{\pe \ : \  \mathbf H_e \pe \leq  \mathbf h_e\}.
\end{equation*}
Let $\bar{\x}_{t} = \hat{\x}_t$ be the nominal state of the fallback trajectory at time $t$ and define the nominal fallback dynamics as
\begin{align*}
    \bar{\x}_{t+1} &= \mathbf A\bar{\x}_t + \mathbf B \bar{\fu}^F_t \\
    \bar{\y}_t &= \mathbf C\bar{\x}_t
\end{align*}
for $t \in \{t, t+T+1\}$ under the open-loop input sequence $\bar\fu_{t:t+T}^F \subset \R^m$. As is normative in robust MPC \cite{MayneRakovic2006}, we consider affine fallback policy sequences that satisfy
\begin{equation}
    \fu^F_{t+k}(\y_t) = \bar\fu^F_{t+k} + \mathbf K (\y_{t+k} - \bar \y_{t+k})
\end{equation}
for $k \in \{0, \dots, T\}$, where $\mathbf K \in \R^{m\times r}$ is a fixed feedback gain specified by a designer. Under this fallback policy sequence, it holds that 
\begin{align*}
    \hat{\x}_{t+1} - \bar{\x}_{t+1} &= (\mathbf A \x_t + \mathbf B\fu^F_t(\y_t) + \w_t - \pe_{t+1}) - (\mathbf A\bar{\x}_t + \mathbf B \bar{\fu}^F_t) \\
    &= \mathbf A (\hat\x_t + \pe_t ) +\mathbf B(\bar \fu_t^F +\mathbf K(\y_t - \bar{\y}_t)) + \w_t - \pe_{t+1} - (\mathbf A\bar{\x}_t +\mathbf B \bar{\fu}^F_t) \\
    &=  \mathbf A (\hat\x_t + \pe_t ) + \mathbf B \mathbf K\mathbf C(\hat{\x}_t + \pe_t - \bar{\x}_t) + \w_t - \pe_{t+1} - \mathbf A\bar{\x}_t \\
    &= (\mathbf A + \mathbf B \mathbf K\mathbf C)(\hat{\x}_t - \bar{\x}_t) + (\mathbf A + \mathbf B\mathbf K\mathbf C)\pe_{t} + \w_t - \pe_{t+1}.
\end{align*}
Therefore, by recursively defining the sets
\begin{align*}
    \calF_0 &:= \{0\} \\
    \calF_{k+1} &:= (\mathbf A + \mathbf B \mathbf K\mathbf C)\calF_k \oplus (\mathbf A + \mathbf B \mathbf K\mathbf C)\calE_\theta \oplus \calW \oplus \calE_{\theta}
\end{align*}
for $k\in \{0,\dots, T\}$ and noting that $\calE_\theta$ is symmetric, we then have that the $k-$step reachable sets under $\fu_{t:t+T}^F$ are given as
\begin{align*}
    \widehat{\calR}_k(\hat{\x}_t, \fu_{t:t+T}^F) &= \{\bar{\x}_{t+k}\} \oplus \calF_k \\
    \calR_k(\hat{\x}_t, \fu_{t:t+T}^F) &=\{\bar{\x}_{t+k}\} \oplus \calF_k \oplus \calE_{\theta}
\end{align*}
for $k \in \{0, \dots, T+1\}$. Therefore, for a linear system, we reformulate \cref{eq:modification-mpc} as 

\begin{equation}\label{eq:linear-mpc}
\begin{aligned}
    \minimize_{\substack{\bar{\fu}^F_{t:t+T|t} \subset \R^m, \\ \fu_{t:t+T|t}\subset \calU}} 
    & \enspace C(\hat{\x}_{t}, \fu_{t:t+T|t}, \bar\fu^F_{t:t+T|t})\\
    \subjectto\enspace & \bar{\x}_{t+k+1|t} = \mathbf A\bar{\x}_{t+k|t} + 
 \mathbf B\bar\fu_{t+k|t}^F \ \quad \quad \forall k \in \{0,\dots, T\},\\
    &\bar{\x}_{t+k|t}  \in \calX \ominus (\calF_k \oplus \calE_\theta)  \quad \quad \quad \quad \quad \  \forall k \in \{0, \dots, T\}, \\
    &\bar\fu^F_{t+k|t} \in \calU \ominus \mathbf K \mathbf C (\calF_k \oplus \calE_\theta) \quad \quad \quad \ \ \forall k \in \{0, \dots, T\}, \\
    & \bar{\x}_{t+T+1|t} \in \calX_R \ominus (\calF_{T+1} \oplus \calE_\theta), \\
    & \fu_{t|t} = \bar\fu^F_{t|t}, \quad \bar{\x}_{t|t} = \hat{\x}_t.
\end{aligned}
\end{equation}
In \cref{eq:linear-mpc}, we specify an objective of the form 
\begin{equation*}
    C(\hat{\x}_{t}, \fu_{t:t+T|t}, \bar\fu^F_{t:t+T|t}) := \sum_{k=0}^T h_N(\bar{\x}^N_{t+k|t}, \fu_{t+k|t}) +  V_N(\bar{\x}_{t+T+1|t}^N), %+h_F(\bar{x}_{t+k|t}, \fu_{t+k|t}^F) +V_F(\bar{x}_{t+T+1|t}^F)
\end{equation*}
where we define the nominal state evolution as following $\bar{\x}_{t|t}^N := \hat{\x}_t$ and $\bar{\x}_{t+k+1}^N = \mathbf A \bar{\x}_{t+k}^N + \mathbf B\fu_{t+k|t}$ for a positive (semi) definite quadratic stage cost $h$ and terminal cost $V$. Moreover, since we assume $\calX$, $\calU$, $\calW$, and $\calE_{\theta}$ are polytopes, we can compute the constraint tightening in \cref{eq:linear-mpc} (that is, $\calX \ominus (\calF_k \oplus \calE_\theta)$, $\calU \ominus \mathbf K \mathbf C (\calF_k \oplus \calE_\theta)$, $\calX_R \ominus (\calF_{T+1} \oplus \calE_\theta)$) via linear programming, rendering \cref{eq:linear-mpc} a convex quadratic program (QP) \cite{GoulartKerrigan2006}. In addition, we use \cite[Alg. 10.5]{BorrelliBemporadEtAl2017} to compute the robust invariant set $\calX_R$ given a recovery policy. 

\section{Approximate solution approach to \cref{eq:modification-mpc} for nonlinear systems with Particle MPC \cite{DyroHarrisonEtAl2021}}\label{ap:mpc-particle}

For nonlinear systems, we propose the use of the Particle MPC (PMPC) algorithm \cite{DyroHarrisonEtAl2021} to approximate a solution to \cref{eq:modification-mpc}. In lieu of explicitly computing the $k-$step reachable sets, the PMPC algorithm accounts for uncertainties by sampling $M\in \mathbb{N}_{>0}$ trajectories from an initial belief state and a given control input sequence, and enforcing the state and input constraints along the $M$ sampled trajectories.  Therefore, when we apply the PMPC algorithm, we consider both open-loop nominal input sequences $\fu_{t:t+T|t} \subset \R^m$ and open-loop fallback input sequences $\fu^F_{t:t+T|t} \subset \R^m$ for simplicity. First,  we define a shorthand dynamics term $\h$ to approximate the dynamics of the perception estimate in \cref{eq:est-dyn} as 
\begin{align*}
    \h(\bar{\x}, \fu, \w, \pe, \pe') &:= \f(\bar{\x} + \pe, \fu, \w) - \pe'. 
\end{align*}
Then, at time $t$,  we sample $M$ disturbance sequences as $\{\w_{t+k|t}^j, \pe_{t+k|t}^j, \pe_{t+k|t}^{'j}\}_{j=0}^M$ from $\calW$ and $\calE_\theta$ so that we can approximate the reachable sets $\widehat\calR_k$. We then approximate \cref{eq:modification-mpc} as

\begin{equation}\label{eq:pmpc}
\begin{aligned}
    \minimize_{\substack{{\fu}^{F}_{t:t+T|t} \subset \calU, \\ \fu_{t:t+T|t}\subset \calU}} 
    & \enspace \frac{1}{M}\sum_{j=0}^M C_j(\hat{\x}_t, \fu_{t:t+T|t}, \fu^F_{t:t+T|t})\\
    \subjectto\enspace & \bar{\x}_{t+k+1|t}^j = \h(\bar{\x}_{t+k|t}^j, \fu^F_{t+k|t}, \w_{t+k|t}, \pe_{t+k|t}^j, \pe_{t+k|t}^{'j}) && \forall k \in \{0,\dots, T\}, \ j\in \{0, \dots, M\},\\
    &\bar{\x}_{t+k|t}^j  \in \calX \ominus \calE_{\theta}, \quad \bar{\x}_{t|t}^j = \hat{\x}_t   && \forall k \in \{0,\dots, T\}, \ j\in \{0, \dots, M\}, \\
    & \bar{\x}_{t+T+1|t}^j \in \calX_R \ominus \calE_\theta,   && \forall j\in \{0, \dots, M\},\\
    & \fu_{t|t} = \fu^F_{t|t},
\end{aligned}
\end{equation}
on which the PMPC algorithm applies sequential convex programming (SCP) to yield a locally optimal solution. In \cref{eq:pmpc}, we have tightened the state constraints further to account for the perception error, i.e., the difference between $\calR$ and $\widehat\calR$. Moreover, in \cref{eq:pmpc}, we optimize the sample average of cost functions $C_j$. We draw the cost functions $C_j$ as a sum of stage costs $h$ over a sampled trajectory under the nominal inputs $\fu_{t:t+T|t}$, i.e., as
\begin{align}\label{eq:pmpccost}
    C_j(\hat{\x}_t, \fu_{t:t+T|t}, \fu^F_{t:t+T|t}) = \sum_{k=0}^T h(\bar{\x}^{N,j}_{t+k|t}, \fu_{t+k|t}) + V(\bar{\x}_{t+T+1|t}^{N,j}),
\end{align}
where $\bar{\x}^{N,j}_{t|t} := \hat{\x}_t$ and $\bar{\x}_{t+k+t}^{N,j} = \h(\bar{\x}_{t+k|t}^j, \fu^F_{t+k|t}, \w_{t+k|t}, \pe_{t+k|t}^j, \pe_{t+k|t}^{'j})$ for $k \in \{0,\dots,  T\}$. We emphasize that the disturbance sequences drawn to compute the cost functions are independent from those we use to evaluate the constraints in \cref{eq:pmpc}, i.e., we sample a total of $2M$ disturbance sequences at each time step, but keep notation minimal in \cref{eq:pmpccost}.

\section{Proof of \cref{lem:conformal}}\label{ap:conformal}
% We first give two preliminary, elementary lemmas that we will use to prove \cref{lem:conformal}. These facts are well-known, but we include them to make this work more self-contained.
% \begin{lemma}
% Let $a_i, b_i$ be real numbers in the range $[0,1]$ for $i = 1,\dots, n$. Then, 
% \begin{equation*}
%     |\prod_{i=1}^n a_i - \prod_{i=1}^n b_i| \leq \sum_{i=1}^n|a_i - b_i|.
% \end{equation*}
% \end{lemma}
% \begin{proof}
% We prove the statement via induction. As a base case, note that 
% \begin{align*}
%     |a_1 a_2 - b_1 b_2| &= |a_1 a_2 + a_1 b_2 - a_1 b_2- b_1 b_2| \\
%     &= |a_1 (a_2 - b_2) - b_2(a_1 - b_2)| \\
%     &\leq a_1 |a_2 - b_2| + b_2|a_1 - b_1| & \text{(triangle inequality)} \\
%     & \leq |a_1 - b_1| + |a_2 - b_2|. & \text{($a_i, b_i \in [0,1]$)}
% \end{align*}
% Then, for the inductive step, note that $a_i, b_i \in [0,1]$ implies that $\prod_i a_i \in [0,1]$, and therefore applying the above logic to $\prod_{i=1}^ka_i$ and $a_{k+1}$ proves the statement.
% \end{proof}

% \begin{lemma}
%     Let $x_i$
% \end{lemma}

% \subsection{Proof of \cref{lem:conformal}}
\conformal*
\begin{proof} First, we note that the false negative rate is bounded by the probability of a missed detection at $t_{\mathrm{stop}}$, conditioned on the event that Algorithm \ref{alg:conformal} does not raise an alarm at any time $t < t_{\mathrm{stop}}(\tau)$, i.e., that is, that $w(a_t) = 0$ for all  $t \in \{0, \dots, t_{\mathrm{stop}}(\tau) - 1\}$. As in the statement and proof of \cref{thm:mpc-safety}, we use the shorthand $t_{\mathrm{fail}} \geq t_{\mathrm{stop}}$ to denote this event, so that
\begin{align}
    \prob_{}\big(\mathrm{False \ Negative}\big) &:= \prob_{} \big(w(a_t) = 0 \ \forall t \in [0, t_{\mathrm{stop}}(\tau)] \ | \ \pe_{t_{\mathrm{stop}}(\tau)} \not \in \calE_\theta \big) \nonumber\\
    &= \prob_{} \big( w(a_{t_{\mathrm{stop}}(\tau)}) = 0 \ | \ \pe_{t_{\mathrm{stop}}(\tau)} \not \in \calE_\theta, \ t_{\mathrm{fail}} \geq t_{\mathrm{stop}})\prob\big( t_{\mathrm{fail}} \geq t_{\mathrm{stop}} \ | \ \pe_{t_{\mathrm{stop}}(\tau)} \not \in \calE_\theta \big) \nonumber \\
    &\leq \prob_{} \big( w(a_{t_{\mathrm{stop}}(\tau)}) = 0 \ | \ \pe_{t_{\mathrm{stop}}(\tau)} \not \in \calE_\theta, \ t_{\mathrm{fail}} \geq t_{\mathrm{stop}}). \label{eq:condprob}
\end{align}

Next, suppose the deployed trajectory $\tau$ conditioned on $t_{\mathrm{fail}} \geq t_{\mathrm{stop}}$ follows the training trajectory distribution $P$ so that $(\tau \ | \ t_{\mathrm{fail}} \geq t_{\mathrm{stop}}), \ \tau_1, \dots, \tau_N \iid P$, and let $\prob_{\mathrm{iid}}$ denote the probability of an event under this assumption. That is, $\prob_{\mathrm{iid}}(\mathrm{False \ Negative})$ denotes the probability of a false negative when $(\tau \ | \ t_{\mathrm{fail}} \geq t_{\mathrm{stop}}), \ \tau_1, \dots, \tau_N \iid P$. 

Under the i.i.d. assumption, it follows that $(\x_{t_{\mathrm{stop}}(\tau)}, \hat{\x}_{t_{\mathrm{stop}}(\tau)}, a_{t_{\mathrm{stop}}(\tau)})$ and the elements of $\calD_{\mathrm{stop}}$ are also i.i.d. 
Therefore, the sequence $(\pe_{t_{\mathrm{stop}}(\tau)}, a_{t_{\mathrm{stop}}(\tau)}), \ (\pe_{t_{\mathrm{stop}}(\tau_1)}^1, a_{t_{\mathrm{stop}}(\tau_1)}^1), \ \dots, \ 
 (\pe_{t_{\mathrm{stop}}(\tau_N)}^N, a_{t_{\mathrm{stop}}(\tau_N)}^N)$ is exchangeable. By \cite[Proposition 1]{LuoZhao2023}, it then follows that 
\begin{equation}\label{eq:confbound}
    \prob_{\mathrm{iid}} \big( w(a_{t_{\mathrm{stop}}(\tau)}) = 0 \ | \ \pe_{t_{\mathrm{stop}}(\tau)} \not \in \calE_\theta , \ t_{\mathrm{fail}} \geq t_{\mathrm{stop}}\big) \leq \delta + \frac{1}{|\calA| + 1},
\end{equation}
because Algorithm \ref{alg:conformal} implements \cite[Algorithm 1]{LuoZhao2023} using $\calD_{\mathrm{stop}}$ as the calibration dataset, $-a_{t_{\mathrm{stop}}}$ as the \emph{surrogate safety score}, and \emph{true safety score} $f(\pe):= \mathbf{1}\{\pe \in \calE_\theta\}$. 

Finally, we bound the difference between the false negative rate under the assumption that $(\tau \ | \ t_{\mathrm{fail}} \geq t_{\mathrm{stop}}), \ \tau_1, \dots, \tau_N \iid P$ and the true false negative rate as
\begin{align}
    &\prob_{} \big( w(a_{t_{\mathrm{stop}}(\tau)}) = 0 \ | \ \pe_{t_{\mathrm{stop}}(\tau)} \not \in \calE_\theta, \ t_{\mathrm{fail}} \geq t_{\mathrm{stop}}) - \prob_{\mathrm{iid}} \big( w(a_{t_{\mathrm{stop}}(\tau)}) = 0 \ | \ \pe_{t_{\mathrm{stop}}(\tau)} \not \in \calE_\theta, \ t_{\mathrm{fail}} \geq t_{\mathrm{stop}}) \nonumber\\
    &\leq \bigg|\prob_{} \big( w(a_{t_{\mathrm{stop}}(\tau)}) = 0 \ | \ \pe_{t_{\mathrm{stop}}(\tau)} \not \in \calE_\theta, \ t_{\mathrm{fail}} \geq t_{\mathrm{stop}}) - \prob_{\mathrm{iid}} \big( w(a_{t_{\mathrm{stop}}(\tau)}) = 0 \ | \ \pe_{t_{\mathrm{stop}}(\tau)} \not \in \calE_\theta, \ t_{\mathrm{fail}} \geq t_{\mathrm{stop}})\bigg| \nonumber \\
    &\leq \mathrm{TV}(P_{t_{\mathrm{stop}}|\mathrm{fault}}, P'_{t_{\mathrm{stop}}|\mathrm{fault}}). \label{eq:tvbound}
\end{align}

% \begin{align*}
%     \prob \big( w(a_{t_{\mathrm{stop}}(\tau)}) = 0 \ | \ \pe_{t_{\mathrm{stop}}(\tau)} \not \in \calE_\theta \big) &= \prob \big( w(a_{t_{\mathrm{stop}}(\tau)}) = 0 \ | \ \pe_{t_{\mathrm{stop}}(\tau)} \not \in \calE_\theta \big) + \prob_{\mathrm{iid}} \big( w(a_{t_{\mathrm{stop}}(\tau)}) = 0 \ | \ \pe_{t_{\mathrm{stop}}(\tau)} \not \in \calE_\theta \big)  \\ & \quad \quad \quad \quad - \prob_{\mathrm{iid}} \big( w(a_{t_{\mathrm{stop}}(\tau)}) = 0 \ | \ \pe_{t_{\mathrm{stop}}(\tau)} \not \in \calE_\theta \big)\\
%     &\leq 
% \end{align*}
Combining \cref{eq:condprob} with \cref{eq:confbound} and \cref{eq:tvbound} then yields
\begin{align*}
        \prob\big(\mathrm{False \ Negative}\big) &\leq \prob_{} \big( w(a_{t_{\mathrm{stop}}(\tau)}) = 0 \ | \ \pe_{t_{\mathrm{stop}}(\tau)} \not \in \calE_\theta, \ t_{\mathrm{fail}} \geq t_{\mathrm{stop}}) \\
        &= \prob_{\mathrm{iid}} \big( w(a_{t_{\mathrm{stop}}(\tau)}) = 0 \ | \ \pe_{t_{\mathrm{stop}}(\tau)} \not \in \calE_\theta, \ t_{\mathrm{fail}} \geq t_{\mathrm{stop}})  \\
        &+ \prob_{} \big( w(a_{t_{\mathrm{stop}}(\tau)}) = 0 \ | \ \pe_{t_{\mathrm{stop}}(\tau)} \not \in \calE_\theta, \ t_{\mathrm{fail}} \geq t_{\mathrm{stop}}) - \prob_{\mathrm{iid}} \big( w(a_{t_{\mathrm{stop}}(\tau)}) = 0 \ | \ \pe_{t_{\mathrm{stop}}(\tau)} \not \in \calE_\theta, \ t_{\mathrm{fail}} \geq t_{\mathrm{stop}}) \\
        % &\leq \delta + \frac{1}{|\calA| + 1} + \bigg|\prob_{} \big( w(a_{t_{\mathrm{stop}}(\tau)}) = 0 \ | \ \pe_{t_{\mathrm{stop}}(\tau)} \not \in \calE_\theta, \ t_{\mathrm{fail}} \geq t_{\mathrm{stop}}) \\
        % &\quad \quad \quad \quad \quad \quad \quad- \prob_{\mathrm{iid}} \big( w(a_{t_{\mathrm{stop}}(\tau)}) = 0 \ | \ \pe_{t_{\mathrm{stop}}(\tau)} \not \in \calE_\theta, \ t_{\mathrm{fail}} \geq t_{\mathrm{stop}})\bigg| \\
        &\leq \delta + \frac{1}{|\calA| + 1} + \mathrm{TV}(P_{t_{\mathrm{stop}}|\mathrm{fault}}, P'_{t_{\mathrm{stop}}|\mathrm{fault}}).
        % \prob_{\mathrm{iid}}\big(\mathrm{False \ Negative}\big) - \prob\big(\mathrm{False \ Negative}\big) + \prob_{\mathrm{iid}}\big(\mathrm{False \ Negative}\big)  \\
        % &\leq \prob_{\mathrm{iid}}\big(\mathrm{False \ Negative}\big) + | \prob\big(\mathrm{False \ Negative}\big) - \prob_{\mathrm{iid}}\big(\mathrm{False \ Negative}\big)|.
\end{align*}

\end{proof}

\section{Proof of \cref{thm:endtoend}}\label{ap:endtoend}
\endtoend*
\begin{proof}
 First, we note that \cref{lem:conformal} bounds the false negative rate of a runtime monitor constructed using Algorithm \ref{alg:conformal}. Therefore, we use \cref{lem:conformal} to lower-bound the probability that the runtime monitor does not miss a detection of a perception fault as 
 \begin{align*}
     \prob(\pe_t \in \calE_{\theta} \ \forall t < t_{\mathrm{fail}}) &= 1 - \prob(\exists t < t_{\mathrm{fail}} \ \mathrm{s.t.} \ \pe_t \not \in \calE_\theta)\\
     &= 1 - \prob(w(a_t) = 0 \ \forall t \in [0, t_{\mathrm{stop}}] \ \cap \  \pe_{t_{\mathrm{stop}}} \not \in \calE_\theta) \\
     &= 1 - \prob(w(a_t) = 0 \ \forall t \in [0, t_{\mathrm{stop}}] \ | \  \pe_{t_{\mathrm{stop}}} \not \in \calE_\theta) \prob(\pe_{t_{\mathrm{stop}}} \not \in \calE_\theta) \\
     &\geq 1 - \prob(w(a_t) = 0 \ \forall t \in [0, t_{\mathrm{stop}}] \ | \  \pe_{t_{\mathrm{stop}}} \not \in \calE_\theta) \\
     &\geq 1 - \delta - \frac{1}{|\calA| + 1} - \mathrm{TV}(P_{t_{\mathrm{stop}}|\mathrm{fault}}, P'_{t_{\mathrm{stop}}|\mathrm{fault}}).
 \end{align*}

Moreover, \cref{thm:mpc-safety} gives us that 
\begin{equation*}
    \prob(\x_t \in \calX, \ \fu_t \in \calU \ \forall t \in [0 , t_{\mathrm{lim}}] \ | \ \pe_t \in \calE_{\theta} \ \forall t < t_{\mathrm{fail}}) = 1.
\end{equation*}
As a shorthand, let $\mathrm{SAFE}$ denote the event that $\x_t \in \calX, \ \fu_t \in \calU \ \forall t \in [0 , t_{\mathrm{lim}}]$.
By applying the law of total probability, it then follows that 
\begin{align*}
    \prob(\mathrm{SAFE}) &\geq \prob(\mathrm{SAFE} \ | \ \pe_t \in \calE_{\theta} \ \forall t < t_{\mathrm{fail}}) \prob(\pe_t \in \calE_{\theta} \ \forall t < t_{\mathrm{fail}}) \\
    &=  \prob(\pe_t \in \calE_{\theta} \ \forall t < t_{\mathrm{fail}}) \\
    &\geq 1 - \delta - \frac{1}{|\calA| + 1} - \mathrm{TV}(P_{t_{\mathrm{stop}}|\mathrm{fault}}, P'_{t_{\mathrm{stop}}|\mathrm{fault}}).
\end{align*}
 
 % Note that by construction, $w(a_{t_{\mathrm{stop}}}) = 1$ at the first event that a perception fault occurs, with high probability by \cref{lem:conformal}. Therefore, $\pe_t \in \calE_{\theta}$ for all $t < t_{\mathrm{fail}}$ with high probability. So, by \cref{thm:mpc-safety}, the safety constraints are satisfied with high probability for all time.  
\end{proof}

\section{Simulation Details}\label{ap:sim}
We include additional details about our simulations in this section. 

\subsection{Planar Quadrotor}
We consider a planar version of the quadrotor dynamics for simplicity, with 2D pose $\bm{p} = [x, y, \theta]^T$, state $\x = [\bm{p}^T, \dot{\bm{p}}^T]^T$, and front and rear input thrust inputs $\fu = [u_f, u_r]^T$ \cite{Tedrake2021}. The disturbance-free dynamics are given as:
\begin{align*}
    \ddot x &= -\frac{1}{m}\sin(\theta)(u_f + u_r) ) \\
    \ddot y &= \frac{1}{m}\cos(\theta)(u_f - u_r) - g \\
    \ddot \theta &= \frac{l}{I}(u_f - u_r)
\end{align*}
We linearize the the dynamics around $\bar{\x} = 0$ and $\bar{\fu} = \frac{mg}{2}[1,1]^T$, and discretize the dynamics using Euler's method with a time step of $\mathrm{dt} =0.15s$. The drone is subject to wind disturbances in the set $\calW = \{\w: (-0.05, -0.02, -1e-3, -1e-3, -1e-3,-1e-3) \leq \w \leq (0.05, 0.02, 1e-3, 1e-3, 1e-3,1e-3)\}$, so that the drone may drift with $\approx 0.33 m / s$ in the $x-$direction without actuation. In our example, the drone has internal sensors to estimate its orientation and velocity, so that $\y = [\theta, \dot{\bm{p}}^T]^T$. The drone estimates its $xy-$position using a hypothetical vision sensor. To do so, we simulate the nominal perception errors as bounded within the set $\calE = \{\pe = \x - \hat{\x} \ : \ \hat{\dot{\bm{p}}} =\bm{p}, \ \hat{\theta} = \theta, \ \|[\hat x,\hat y] - [x,y]\|_{\infty} \leq 0.05 \}$, so that the perception system estimates the $xy-$position within a $10 \mathrm{cm}$-wide box around the true position nominally, and perfectly outputs $\y$. When the vision system fails, we randomly sample the $xy-$position within the range $(-10,10)$. We set the time limit $t_{\mathrm{lim}} = 8 \mathrm{s}$, and the controller horizon to $T=10$. In these simulations we give the Fallback-Safe MPC a perfect runtime monitor, so that $w(\cdot) = 1 - \mathbbm{1}\{\pe_t \in \calE\}$. We implement the Fallback-Safe MPC using the tube MPC formulation in \cref{ap:mpc-linear}, where we place a quadratic distance penalty $P = I_{6\times6}$ to a goal location on the nominal trajectory and a quadratic cost $R = I_{2\times2}$ on each input.

First, in \cref{fig:landing}, we simulate a scenario where the drone attempts a vision-based landing at the origin. Here, the state constraint is not to crash into the ground: $\calX := \{\x \ :\ y \geq 0\}$. The fallback is for the drone to abort the landing and fly away at a sufficient altitude. So, we take the recovery set as $\calX_R = \{\x \ : \ \y \geq 2\}$. We note that under the recovery policy $\pi_R(\y):= \epsilon + K \y$, where $K$ stabilizes the orientation $\theta$ around $0$ and $\epsilon$ is a small offset to continually fly upward, the recovery set is invariant under both the estimator dynamics \cref{eq:est-dyn} and the state dynamics \cref{eq:dynamics}, so the Fallback-Safe MPC is recursively feasible by \cref{thm:mpc-safety}. We compare the Fallback-Safe MPC with a naive tube MPC that optimizes only a single trajectory and does not anticipate vision failures. As shown in \cref{fig:landing}, the Fallback-Safe MPC plans fallback trajectories that safely abort the landing and fly away into open space. In contrast, the naive tube MPC does not reason about perception faults, and crashes badly when the perception fails starting at $t = 10\mathrm{dt}$. Therefore, this example demonstrates the necessity of planning with a fallback. 
% We also see that the perception uncertainty prevents the Fallback-Safe MPC from completing the landing to ensure the true state does not leave $\calX$.

% \begin{figure}[t]
%     \centering
%     \includegraphics[width=\textwidth]{figures/drone_figures/fallback-safe-vs-naive-1.pdf}
%     \caption{Trajectories of the planar quadrotor in the $xy$-plane. The realized quadrotor trajectories are in black, and the icons show the orientation of the quadrotor at every $k=7$ time steps. The safe recovery region is highlighted in green. The blue-dashed line indicate the state constraint. Left: In red, we plot the predicted reachable sets of the fallback strategy and in blue, we plot the predicted nominal trajectories, both at $k=7$ step intervals. Right: In blue, we plot the predicted reachable tubes of the Naive Tube MPC at $k$ step intervals. \vspace{-1em}}
%     \label{fig:landing}
% \end{figure}
% \begin{figure}[b]
%     \centering
%     \includegraphics[width=\textwidth]{figures/drone_figures/unsafe-vs-safe-recovery-2.pdf}
%     \caption{Trajectories of the planar quadrotor in the $xy$-plane. The disjoint safe recovery regions are highlighted in green, the unsafe ground region is the unmarked region $|x|<1, \ y\leq0$. Both the top and bottom figure follow the layout in \cref{fig:landing} (left).  }
%     \label{fig:goal_loc}
% \end{figure}

Secondly, we simulate a scenario where the drone must navigate towards the in-air $xy$ goal location $x_g = [-8,3]$ the box state constraint $\calX = \{\x \ : \ -10 \leq x \leq 10,\ 0 \leq y \leq 4\}$. Here, when the drone loses its vision, it is no longer possible to avoid the boundaries of $\calX$ using only the fallback measurement $\y$. Instead, as in the example in \cref{fig:hero-fig}, our recovery set is to land the drone. To model the drone as having landed, we modify the dynamics to freeze the state for all remaining time once the state $\x$ enters the $y \leq 0$ region with $|\theta| \leq \pi/12$ and $\|\dot{p}\|_{\infty} \leq .1$. However, in this example, the drone must cross an unsafe ground region, such as the busy road in the example in \cref{fig:hero-fig}, specified as the region of states with $|x| < 1, \ y\leq 0$. Therefore, we take the recovery region $\calX_R$ as all landed states with $|x|\geq 1$. 

Clearly, $\calX_R$ is a safe recovery region under $\pi_R(\y) = 0$ under the true dynamics \cref{eq:dynamics}. We note that in this example, $\calX_R$ is not RPI under the state estimate dynamics \cref{eq:est-dyn} in nominal conditions, because the estimation error bound allows $\hat{\x}$ to leave the $\calX_R$ even if $\x \in \calX_R$. To retain the safety guarantee, we also trigger the fallback if \cref{eq:modification-mpc} is infeasible, a simple fix first proposed in \cite{KollerBerkenkampEtAl2018}. We did not observe recursive feasibility issues in the simulations. For the drone to cross the road, we need to maintain recoverability with two disjoint recovery regions. Rather than add a mixed-integer constraint to \cref{eq:modification-mpc}, we solve multiple versions of \cref{eq:modification-mpc} at each time step, one for each recovery region, and select the feasible input with lowest cost. We compare our approach with another naive baseline, which we label the Unsafe Recovery MPC, that executes a nominal MPC policy and naively tries to compute a fallback trajectory post-hoc, using the previous estimate before the fault occurred. As shown in \cref{fig:goal_loc}, our Fallback-Safe MPC first maintains feasibility of the fallback with respect to the rightmost recovery region, slows down, and then switches to the leftmost recovery region once a feasible trajectory is found. In contrast, the Unsafe Recovery MPC does not maintain the feasibility of the fallback by modifying nominal operations, and is forced to crash land in the unsafe ground region (rather than throwing an infeasibility error, our implementation relies on slack variables). Therefore, this example illustrates that it is necessary to modify nominal operations to maintain the feasibility of a fallback.

\subsection{X-Plane Aircraft Simulator}
Finally, we evaluate the conformal prediction Algorithm \ref{alg:conformal} and the end-to-end safety guarantee of our framework using the photo-realistic X-Plane 11 simulator. We simulate an autonomous aircraft taxiing down a runway with constant reference velocity, while using a DNN to estimate its heading error (HE) $\theta$ and center-line distance (cross-track error (CTE)) $y$ from an outboard camera feed. Here, internal encoders always correctly output the velocity $v$, so that $\x = [x, y, \theta, v]$ and $\y = v$. The aircraft must not leave the runway, given by the state constraint $|y|\leq 6 \mathrm{m}$. For control, we model the taxiing aircraft as a unicycle, with disturbance-free dynamics following
\begin{align*}
    \dot{x} &= v \cos(\theta) \\
    \dot{y} &= v \sin(\theta) \\
    \dot{\theta} &= u_\theta \\
    \dot{v} &= c_1(u_v - c),
\end{align*}
where $u_\theta$ is the steering command and $c, \ c_1 \in \R$ are constants so that the 
 acceleration command $u_v = 0$ implies maximum braking and $u_v=1$ implies maximum acceleration. We Euler-discretize the dynamics at a timestep $\mathrm{dt} = 1\mathrm{s}$, and control the simulation using both open-loop input sequences for both the nominal trajectory and the fallback trajectory with a horizon of $T=5\mathrm{s}$ and a identity quadratic costs on the tracking performance and actuation, around a reference speed of $5 \mathrm{m/s}$. 

 We train the DNN perception model on $4\times10^4$ labeled images collected only in morning, clear sky weather, but we deploy the system in a context $P_\rho$ where it may experience a variety of weather conditions (depicted in \cref{fig:xplane}). 
Following other works using an earlier version of the X-Plane 11 simulator \cite{KatzCorso2022, KatzCorsoEtAl2021}, we use a simple multi layer perceptron (MLP) as the perception model with a least-squares loss. Specifically, we downsample and grayscale the input image to a $256\times 128$ pixel image and use a 4-layer MLP with 256 hidden units per layer (and 2 outputs). We found as simple architecture like the MLP to be sufficient for the centerline estimation task, with little gain from more complex architectures like CNNs on the training data. 

We parameterize the environment $\rho := (\mathrm{weather \ type}, t_{\mathrm{start}}, \mathrm{severity})$ as a triplet indicating the weather type, severity level, and starting time from which the visibility starts to degrade, so that in the deployment context $P_{\rho}$, we randomly sample an environment that starts with clear-skies and high visibility, but may cause OOD errors during an episode. The corruption types we simulate are:
\begin{center}
\begin{enumerate}
    \item Night-time darkness 
    \item Motion blur
    \item Gaussian noise
    \item Rain
    \item Rain and motion blur
    \item Snow on tarmac
    \item Snowing and snow on tarmac
\end{enumerate}
\end{center}
To sample an environment, we first flip a biased coin so that we decide to sample an environment with OOD weather with probability $1/3$ and experience no OOD scenario with probability $2/3$. We then randomly select one of the corruption types and a severity level in the range $[1,5]$, where $5$ almost completely blocks visibility, and $1$ is a slight difference. Finally, we sample the start time after which the perception starts to degrade in the range $13-19$ seconds. After the start time of the OOD event, we linearly ramp the severity of the corruption from $0$ to the chosen severity level within 5 seconds. 

As shown in \cref{fig:example}, heavy weather degrades the perception significantly. The fallback is to brake the aircraft to a stop, so that the recovery set is given as $\calX_R = \{\x \ : \ v =0, \ |y|\leq 6\}$, which is invariant under $\pi_R(\y) := 0$ (see footnote \ref{fn:infeas}). As in \cite{RichterRoy2017}, we train an autoencoder alongside the perception DNN on the morning, clear sky data and use the least-squares reconstruction error as the anomaly signal $a_t$. For the anomaly detector, we use a symmetric MLP autoencoder with 3 layers of 128 hidden units each and set the dimension of the latent space to 64. We define the perception error set as $\calE = \{\pe = \hat{\x} - \x \ :\ |\hat{\theta} - \theta | \leq 7^{\circ}, \ |\hat{y} - y|\leq 1.3 \mathrm{m}\}$, to include at most $7$ degree HE and $1.3\mathrm{m}$ CTE.

We record 100 training trajectories using the Fallback-Safe MPC \cref{eq:modification-mpc} and a ground-truth supervisor $w(\cdot) = 1 - \mathbbm{1}\{\pe_t \in \calE\}$ to calibrate Algorithm \ref{alg:conformal}, and then evaluate on 900 test trajectories with environments sampled i.i.d. from $P_\rho$ using and Algorithms \ref{alg:fallback-safe-mpc}-\ref{alg:conformal}. This ensures we satisfy \cref{def:safety} by corollary \cref{cor:endtoend}. We compute the empirical false positive and false negative rate when we evaluate Algorithm \ref{alg:conformal} with various values of $\delta \in [0,1]$ in \cref{fig:xplane-results} (left). As \cref{fig:xplane-results} (left) shows, the FNR of Algorithm \ref{alg:conformal} satisfies \cref{lem:conformal}'s $\delta' = \delta + 1 / (|\calA| + 1)$ bound for all values of $\delta$, validating our guarantees. Moreover, the false positive rate, the rate at which we incorrectly trigger the fallback, is near $0$ for risk tolerances as small as $\delta' = 5\%$. This shows that algorithm \ref{alg:conformal} is highly sample efficient and not overly conservative, since it hardly issues incorrect alarms. In \cref{fig:xplane-results} (right), we control the system with an end-to-end safety guarantee of $\delta' = .1$ using our framework and observe no constraint violations. For the trajectories in which we triggered the fallback, \cref{fig:xplane-results} shows that over $80\%$ would have led to an aircraft failure had we not interfered. This shows that our framework is effective at avoiding robot failures, and experiences few unnecessary interruptions with an effective OOD detection heuristic like the autoencoder reconstruction loss.
%% 1. use algorithm 1
%% make a guarantee using corrolary end-to-end
%% Xplane 11 simulator
%% Aircraft taxiing down a runway
%% Cross-track error < 6m
%% DNN is trained on morning, clear sky weather. Evaluate on mixture of no weather, and sudden shifts to darkness, rain, snow, which derail the DNN predictions
%% Use an autoencoder reconstruction loss to do OOD detection
%% The fallback is to brake, and stop the aircraft as soon as possible. 
%% 
%% record 20 training trajectories {replace with 200}, and evaluate the predictor on 80 {replace with 1800} test trajectories. time limit is 60 seconds
%% 
\end{appendices}

\end{document}